\documentclass[letterpaper,11pt]{article}
\usepackage{style}
\usepackage{dsfont}
\usepackage[utf8]{inputenc}
\usepackage[margin=1in,footskip=0.25in]{geometry}
\begin{document}

\title{Exact Community Recovery in the Geometric SBM}
\author{Julia Gaudio\thanks{(\url{julia.gaudio@northwestern.edu}) Department of Industrial Engineering and Management Sciences, Northwestern
University} \and Xiaochun Niu\thanks{(\url{xiaochunniu2024@u.northwestern.edu}) Department of Industrial Engineering and Management Sciences, Northwestern
University} \and Ermin Wei\thanks{(\url{ermin.wei@northwestern.edu}) Department of Electrical and Computer Engineering and Department of Industrial Engineering and Management Sciences, Northwestern
University}}
\date{}
\maketitle

\begin{abstract} We study the problem of exact community recovery in the Geometric Stochastic Block Model (GSBM), where each vertex has an unknown community label as well as a known position, generated according to a Poisson point process in $\mathbb{R}^d$. Edges are formed independently conditioned on the community labels and positions, where vertices may only be connected by an edge if they are within a prescribed distance of each other. The GSBM thus favors the formation of dense local subgraphs, which commonly occur in real-world networks, a property that makes the GSBM qualitatively very different from the standard Stochastic Block Model (SBM). We propose a linear-time algorithm for exact community recovery, which succeeds down to the information-theoretic threshold, confirming a conjecture of Abbe, Baccelli, and Sankararaman. The algorithm involves two phases. The first phase exploits the density of local subgraphs to propagate estimated community labels among sufficiently occupied subregions, and produces an almost-exact vertex labeling. The second phase then refines the initial labels using a Poisson testing procedure. Thus, the GSBM enjoys \emph{local to global amplification} just as the SBM, with the advantage of admitting an information-theoretically optimal, linear-time algorithm.\end{abstract}

\section{Introduction}

  Community detection is the problem of identifying latent community structure in a network. In 1983, Holland, Laskey, and Leinhardt \cite{Holland1983} introduced the \emph{Stochastic Block Model} (SBM), a probabilistic model which generates graphs with community structure, where edges are generated independently conditioned on community labels.
  Since then, the SBM has been intensively studied in the probability, statistics, machine learning, and information theory communities. Many community recovery problems are now well-understood; for example, the fundamental limits of the exact recovery problem are known, and there is a corresponding efficient algorithm that achieves those limits \cite{Abbe2015}. For an overview of theoretical developments and open questions, please see the survey of Abbe \cite{Abbe2017}. 
  
  While the SBM is a powerful model, its simplicity fails to capture certain properties that occur in real-world networks. In particular, social networks typically contain many triangles; a given pair of people are more likely to be friends if they already have a friend in common \cite{rapoport1953spread}. The SBM by its very nature does not capture this transitive behavior, since edges are formed independently, conditioned on the community assignments. To address this shortcoming, Baccelli and Sankararaman \cite{Sankararaman2018} introduced a spatial random graph model, which we refer to as the Geometric Stochastic Block Model (GSBM). In the GSBM, vertices are generated according to a Poisson point process in a bounded region of $\mathbb{R}^d$. Each vertex is randomly assigned one of two community labels, with equal probability. A given pair of vertices $(u,v)$ is connected by an edge with a probability that depends on  both the community labels of $u$ and $v$ as well as their distance. Edges are formed independently, conditioned on the community assignments and locations. The geometric embedding thus governs the transitive edge behavior. The goal is to determine the communities of the vertices, observing the edges and the locations. In a follow-up work, Abbe, Sankararaman, and Baccelli \cite{Abbe2021} studied both partial recovery in sparse graphs, as well as exact recovery in logarithmic-degree graphs. Their work established a phase transition for both partial and exact recovery, in terms of the Poisson intensity parameter $\lambda$. The critical value of $\lambda$ was identified in some special cases of the sparse model, but a precise characterization of the information-theoretic threshold for exact recovery in the logarithmic regime was left open.
  
  Our work resolves this gap, by identifying the information-theoretic threshold for exact recovery in the logarithmic degree regime (and confirming a conjecture of Abbe et al \cite{Abbe2021}). Additionally, we propose a polynomial-time algorithm achieving the information-theoretic threshold. The algorithm consists of two phases: the first phase produces a preliminary almost-exact labeling through a local label propagation scheme, while the second phase refines the initial labels to achieve exact recovery. At a high level, the algorithm bears some similarity to prior works on the SBM using a two-phase approach \cite{Abbe2015,Mossel2015}. Our work therefore shows that just like the SBM, the GSBM exhibits the so-called \emph{local to global amplification} phenomenon \cite{Abbe2017}, meaning that exact recovery is achievable whenever the probability of misclassifying an individual vertex, given the labels of the remaining $n-1$ vertices, is $o(1/n)$. However, the GSBM is qualitatively very different from the SBM, and it is not apparent at the outset that it should exhibit local to global amplification. In particular, the GSBM is not a low-rank model, suggesting that approaches such as spectral methods \cite{Abbe2021} and semidefinite programming \cite{Hajek2016}, which exploit the low-rank structure of the SBM, may fail in the GSBM. In order to achieve almost exact recovery in the GSBM, we instead use the density of local subgraphs to propagate labels. Our propagation scheme allows us to achieve almost exact recovery, and also ensures that no local region has too many misclassified vertices. The dispersion of errors is crucial to showing that labels can be correctly refined in the second phase.
  
  Notably, our algorithm runs in linear time (where the input size is the number of edges). This is in contrast with the SBM, for which no statistically optimal linear-time algorithm for exact recovery has been proposed. To our knowledge, the best-known runtime for the SBM in the logarithmic degree regime is achieved by the spectral algorithm of Abbe et al \cite{Abbe2020}, which runs in $O(n \log^2 n)$ time, while the number of edges is $\Theta(n \log n)$. More recent work of Cohen--Addad et al \cite{Cohen2022} proposed a linear-time algorithm for the SBM, but the algorithm was not shown to achieve the information-theoretic threshold for exact recovery. Intuitively, the strong local interactions in the GSBM enable more efficient algorithms than what seems to be possible in the SBM. 
  
  \vskip6pt
  \paragraph*{Notation and organization.} We write $[n]=\{1,\cdots, n\}$. We use Bachmann--Landau notation with respect to the parameter $n$; i.e. $o(1)$ means $o_n(1)$. $\text{Bin}$ denotes the binomial distribution. For $\mu \in \mathbb{R}^m$, $\text{Poisson}(\mu)$ denotes the $m$-type Poisson distribution. 
  
  The rest of the paper is organized as follows. Section \ref{sec:results} describes the exact recovery problem as well as our main result (Theorem \ref{theorem:exact-recovery}). The exact recovery algorithm is given in Section \ref{sec:algorithm}, along with an outline of the proof of exact recovery. Sections \ref{sec:proof-phase-I} and \ref{sec:proof-phase-II} include the proofs of the two phases of the algorithm. Section \ref{sec:impossibility} contains the proof of impossibility (Theorem \ref{theorem:impossibility-general}) (a slight generalization of \cite[Theorem 3.7]{Abbe2021} to cover the disassortative case). Section \ref{sec:related-work} includes additional related work. We conclude with future directions in Section \ref{sec:future-directions}.
  
  \section{Model and main results}\label{sec:results}
  We now describe the GSBM in the logarithmic degree regime, where edges are formed only between sufficiently close vertices, as proposed in \cite{Sankararaman2018,Abbe2021}. 
  \begin{definition}\label{def:gsbm}
  Let $\lambda > 0$, $a,b\in[0,1]$, and $a\neq b$ be constants, and let $d \in \mathbb{N}$. A graph $G$ is sampled from $\text{GSBM}(\lambda, n, a, b, d)$ according to the following steps:
  \begin{enumerate}
      \item The locations of vertices are determined according to a homogeneous Poisson point process\footnote{The definition and construction of a homogeneous Poisson point process are provided in Definition \ref{def:PPP}.} with intensity $\lambda$ in the region $\cS_{d,n} := [-n^{1/d}/2, n^{1/d}/2 ]^d \subset \mathbb{R}^d$. Let $V\subset\cS_{d,n}$ denote the vertex set. \label{sample:step-1}
      \item Community labels are generated independently. The ground truth label of vertex $u \in V$ is given by $\sigma_0(u) \in \{-1, 1\}$, with $\mathbb{P}(\sigma_0(u) = 1) = \mathbb{P}(\sigma_0(u) = -1) = 1/2$.
      \item Conditioned on the locations and community labels, edges are formed independently. Letting $E$ denote the edge set, for $u, v \in V$ and $u \neq v$, we have
      \begin{align*}
      \mathbb{P}(\{u,v\} \in E) &= \begin{cases}
      a & \text{if } \sigma_0(u) = \sigma_0(v), \ \Vert u - v \Vert \leq (\log n)^{1/d}\\
      b & \text{if } \sigma_0(u) \neq \sigma_0(v), \ \Vert u - v \Vert \leq (\log n)^{1/d}\\
      0 & \text{if } \Vert u - v \Vert > (\log n)^{1/d}.
      \end{cases}
      \end{align*}
  \end{enumerate}
  The graph does not contain self-loops.
      Here $\Vert u - v \Vert$ denotes the toroidal metric:
      \[\Vert u - v \Vert = \big \Vert \min \big\{|u_i - v_i|, n^{1/d} - |u_i -v_i| \big\}, 
   \dots, \min \big\{|u_d - v_d|, n^{1/d} - |u_d -v_d| \big\}\big\Vert_2,\]
      where $\Vert \cdot \Vert_2$ is the standard Euclidean metric.
  \end{definition}
  In other words, a given pair of vertices can only be connected by an edge if they are within a distance of $(\log n)^{1/d}$; in that case, we say they are \emph{mutually visible}. When a pair of vertices are mutually visible, the probability of being connected by an edge depends on their community labels, as in the standard SBM. Observe that any unit volume region has $\text{Poisson}(\lambda)$ vertices (and hence $\lambda$ vertices in expectation). In particular, the expected number of vertices in the region $\mathcal{S}_{d,n}$ is $\lambda n$. 
  
  Given an estimator $\widetilde{\sigma} = \widetilde{\sigma}_n$, we define $A(\widetilde{\sigma}, \sigma_0) = \max_{s\in\{\pm 1\}}(\sum_{u\in V}\mathds{1}_{\widetilde\sigma(u) = s\sigma_0(u)})/|V|$ as the agreement of $\widetilde{\sigma}$ and $\sigma_0$. We define some recovery requirements including \emph{exact recovery} as follows.
  \begin{itemize}
      \item \emph{Exact recovery:} $\lim\limits_{n \to \infty} \pr(A(\widetilde{\sigma}, \sigma_0)=1) = 1$,
      \item \emph{Almost exact recovery:} $\lim\limits_{n \to \infty} \pr(A(\widetilde{\sigma}, \sigma_0)\ge1-\epsilon) = 1$, for all $\epsilon>0$,
      \item \emph{Partial recovery:} $\lim\limits_{n \to \infty} \pr(A(\widetilde{\sigma}, \sigma_0)\ge\alpha) = 1$, for some $\alpha >1/2$.
  \end{itemize}
  In other words, an exact recovery estimator must recover all labels (up to a global sign flip), with probability tending to $1$ as the graph size goes to infinity. 
  Abbe et al \cite{Abbe2021} identified an impossibility regime for the exact recovery problem. Here, $\nu_d$ is the volume of a unit Euclidean ball in $d$ dimensions.
  \begin{theorem}[Theorem 3.7 in \cite{Abbe2021}]\label{theorem:impossibility}
  Let $\lambda > 0$, $d \in \mathbb{N}$, and $0 \leq b < a \leq 1$ satisfy 
  \begin{equation}
  \lambda \nu_d (1-\sqrt{ab} - \sqrt{(1-a)(1-b)} ) < 1, \label{eq:IT-threshold}    
  \end{equation}
  and let $G \sim \text{GSBM}(\lambda, n, a, b, d)$. Then any estimator $\widetilde{\sigma}$ fails to achieve exact recovery.
  \end{theorem}
  Abbe et al \cite{Abbe2021} conjectured that the above result is tight, but only established that exact recovery is achievable for $\lambda > \lambda(a,b,d)$ sufficiently large \cite[Theorem 3.9]{Abbe2021}. In this regime, \cite{Abbe2021} provided a polynomial-time algorithm based on the observation that the relative community labels of two nearby vertices can be determined with high accuracy by counting their common neighbors. By taking $\lambda > 0$ large enough to drive up the density of points, the failure probability of pairwise classification can be taken to be an arbitrarily small inverse polynomial in $n$. 
  
  Our main result is a positive resolution to \cite[Conjecture 3.8]{Abbe2021} (with a slight modification for the case $d=1$, noting that $\nu_1=2$). 
  \begin{theorem}[Achievability]\label{theorem:exact-recovery}
  There exists a polynomial-time algorithm achieving exact recovery in $G \sim \text{GSBM}(\lambda, n, a, b, d)$ whenever
  \begin{enumerate}
      \item $d = 1$, $\lambda>1$, $a,b \in [0,1]$, and $ 2\lambda (1-\sqrt{ab} - \sqrt{(1-a)(1-b)} ) > 1$; or
      \item $d \geq 2$, $a,b \in [0,1]$, and $\lambda \nu_d (1-\sqrt{ab} - \sqrt{(1-a)(1-b)} ) > 1$.
  \end{enumerate}
  \end{theorem}
  We  drop the requirement that $a > b$ in Theorem \ref{theorem:impossibility}, thus covering the disassortative case. We additionally expand the impossible regime for $d=1$, compared to Theorem \ref{theorem:impossibility}.
  \begin{theorem}[Impossibility]\label{theorem:impossibility-general}
  Let $\lambda > 0$, $d \in \mathbb{N}$, and $a,b \in [0,1]$ satisfy \eqref{eq:IT-threshold} and let $G \sim \text{GSBM}(\lambda, n, a, b, d)$. Then any estimator $\widetilde{\sigma}$ fails to achieve exact recovery. Additionally, if $d = 1$ and $\lambda < 1$, then any estimator $\widetilde{\sigma}$ fails to achieve exact recovery.
  \end{theorem}
  Putting Theorems \ref{theorem:exact-recovery} and \ref{theorem:impossibility-general} together establishes the information-theoretic threshold for exact recovery in the GSBM, and shows that recovery is efficiently achievable above the threshold. We remark that the condition $\lambda \nu_d(1 - \sqrt{ab} - \sqrt{(1-a)(1-b)}) > 1$ in Theorem \ref{theorem:exact-recovery} is equivalent to $D_+(x \| y) > 1$, where $D_+(x \| y)$ is the Chernoff--Hellinger (CH) divergence \cite{Abbe2015} between the vectors $x = \lambda \nu_d [a, 1-a,b,1-b ]/2$ and $y = \lambda \nu_d [b, 1-b,a,1-a ]/2$. As we will show, the exact recovery problem can be reduced to a multitype Poisson hypothesis testing problem; the CH-divergence condition characterizes the parameters for which the hypothesis test is successful.
  
  Abbe et al \cite{Abbe2021} suggested that the threshold given by Theorem \ref{theorem:impossibility} might be achieved by a two-round procedure reminiscent of the exact recovery algorithm for the SBM developed by Abbe and Sandon \cite{Abbe2015}. Indeed, our algorithm is a two-round procedure, but the details of the first phase (achieving almost exact recovery) are qualitatively very different from the strategy employed in the standard SBM. At a high level, our algorithm spreads vertex label information locally by exploiting the density of local subgraphs. The information is spread by iteratively labeling ``blocks'', labeling a given block by using a previously labeled block as a reference. To ensure that the algorithm spreads label information to all (sufficiently dense) blocks, we establish a connectivity property of the dense blocks that holds with high probability whenever $\lambda \nu_d > 1$ ($\lambda > 1$ if $d = 1$). This is in contrast to the Sphere Comparison algorithm \cite{Abbe2015} for the SBM, where the relative labels of a pair of vertices $u,v$ are determined by comparing their neighborhoods.
  
  The algorithm in Phase I in fact achieves almost exact recovery for a wider range of parameters than what is required to achieve exact recovery. 
  \begin{theorem}\label{theorem:almost-exact-recovery}
  There is a polynomial-time algorithm achieving almost exact recovery in $G \sim \text{GSBM}(\lambda, n, a, b, d)$ whenever
  \begin{enumerate}
      \item $d = 1$, $\lambda > 1$, and $a, b \in [0,1]$ with $a \neq b$; or
      \item $d \geq 2$, $\lambda \nu_d > 1$, and $a, b \in [0,1]$ with $a \neq b$.
  \end{enumerate}
  \end{theorem}
  
  \section{Exact recovery algorithm}\label{sec:algorithm}

  This section presents our algorithm, which consists of two phases. In Phase I, our goal is to estimate an almost-exact labeling $\widehat{\sigma} \colon V \to \{-1, 0,1\}$, where the label $0$ indicates uncertainty. Phase I is based on the following observation: for any $\delta > 0$, if we know the true labels of some $\delta \log n$ vertices visible to a given vertex $v$, then by computing edge statistics, we can determine the label of $v$ with probability $1-n^{-c}$, for some $c(\delta)> 0$. In Phase I, we partition the region into hypercubes of volume $\Theta (\log n)$ (called \emph{blocks}), and show how to produce an almost exact labeling of all blocks that contain at least $\delta \log n$ vertices (called \emph{occupied blocks}), by an iterative label propagation scheme. Next, Phase II refines the labeling $\widehat{\sigma}$ to $\widetilde\sigma$ using Poisson testing. Phase II builds upon a well-established approach in the SBM literature \cite{Abbe2015, Mossel2015}, to refine an almost-exact labeling with dispersed errors into an exact labeling. The main novelty of our algorithm therefore lies in Phase I.
  
  Before describing the algorithm, we introduce the notion of a \emph{degree profile}.
  \begin{definition}[Degree profile]
      Given $G\sim \text{GSBM}(\lambda, n, a, b, d)$, the \emph{degree profile} of a vertex $u \in V$ with respect to a reference set $S \subset V$ and a labeling $\sigma \colon S \to \{-1, 1\}$ is given by the $4$-tuple, 
      \$ d(u, \sigma, S) = \big[d_1^+(u, \sigma, S), d_1^-(u, \sigma, S), d_{-1}^{+}(u, \sigma, S), d_{-1}^-(u, \sigma, S)\big], \$
      where
      \$
      &d_{1}^+(u, \sigma, S) = |\{v \in S \colon \sigma(v) = 1, (u,v) \in E \}|, \\ 
      &d_{1}^-(u, \sigma, S) = |\{v \in S \colon  \sigma(v) = 1, (u,v) \not \in E, \Vert u - v \Vert \leq (\log n)^{{1}/{d}}\} |,\\
      & d_{-1}^+(u, \sigma, S) = |\{v \in S\colon  \sigma(v) = -1, (u,v) \in E \}|,  \\
      &d_{-1}^-(u, \sigma, S) = |\{v \in S\colon \sigma(v) = -1, (u,v) \not \in E, \Vert u - v \Vert \leq (\log n)^{{1}/{d}}\} |.
      \$
      Note that we only consider $v$ such that $\Vert u - v \Vert \leq (\log n)^{1/d}$, since we only want to count non-edges to vertices that are visible to $u$. For convenience, when $V$ serves as the reference set, we write $d(u, \sigma) := d(u, \sigma, V)$ and $d(u, \sigma) := [d_1^+(u, \sigma), d_1^-(u, \sigma), d_{-1}^{+}(u, \sigma),d_{-1}^-(u, \sigma)]$. 
  \end{definition}
  
  \subsection{Exact recovery for \texorpdfstring{$d=1$}{}.}
  We first describe the algorithm specialized to the case $d=1$. Several additional ideas are required to move to the $d \geq 2$ case, to ensure uninterrupted propagation of label estimates over all occupied blocks. We first describe the simplest case where $d=1$, $\lambda >2$, and $a,b \in [0,1]$ with $a \neq b$.
  
  \vskip6pt
  \paragraph*{Algorithm for $\lambda > 2$.} 
  The algorithm is presented in Algorithm \ref{alg:almost-exact-eg}. 
  In Phase I, we first partition the interval into blocks of length ${\log n}/{2}$ and define $V_i$ as the set of vertices in the $i$th block for $i\in [2n/\log n]$. In this way, any pair of vertices in adjacent blocks are within a distance of $\log n$.
  The density $\lambda>2$ ensures a high probability that all blocks have $\Omega(\log n)$  
  vertices, as we later show in \eqref{eq:connected-H-lambda2}. 
  Next, we use the \texttt{Pairwise Classify} subroutine to label the first block (Line \ref{line:step2-eg}). Here, we select an arbitrary vertex $u_0 \in V_1$ and set $\widehat{\sigma}(u_0) = 1$. The labels of other vertices $u \in V_1$ are labeled by counting common neighbors with $u_0$, among the vertices in $V_1$.
  Next, the labeling of $V_1$ is propagated to other blocks $V_i$ for $i\ge 2$ utilizing the edges between $V_{i-1}$ and $V_i$ and the estimated labeling on $V_{i-1}$, by thresholding degree profiles with respect to $V_{i-1}$ according to Algorithm \ref{alg:propagation} (Lines \ref{line:step3-eg}-\ref{line:step3-end-eg}). The reference set $S$ in Algorithm \ref{alg:propagation} plays the role of $V_{i-1}$ and $S'$ plays the role of $V_i$. 
  Intuitively, if $a>b$, a vertex tends to exhibit more edges and fewer non-edges within its own community while having fewer edges and more non-edges with the other community. Conversely, if $a < b$, the opposite observation holds.
  In order to classify the vertices in $V_i$, we use edges from $V_i$ to the larger set of $\{u \in V_{i-1}: \widehat{\sigma}(u) = 1\}$ and $\{u \in V_{i-1}: \widehat{\sigma}(u) = -1\}$, rather than using all edges between $V_i$ and $V_{i-1}$, which simplifies the analysis. 
  In Theorem \ref{thm:phase1-summary}, we will demonstrate that Phase I
  achieves almost-exact recovery on $G$ under the conditions in Theorem \ref{theorem:almost-exact-recovery}.
  
  \begin{breakablealgorithm}
      \caption{Exact recovery for the GSBM ($d=1$ and $\lambda>2$)} \label{alg:almost-exact-eg}
      \begin{algorithmic}[1]
      \Require{$G \sim \text{GSBM}(\lambda, n,a,b,1)$ where $\lambda > 2$.}
      \Ensure{An estimated community labeling $\widetilde{\sigma}: V \to \{-1,1\}$.}
      \vspace{5pt}
      \State{{\bf Phase I:}} 
      \State{Partition the interval $[-n/2, n/2]$ into $2n/\log n$ blocks\footnotemark \ of volume $\log n/2$ each. Let $B_i$ be the $i$th block and $V_i$ be the set of vertices in $B_i$ for $i\in [2n/\log n]$.} \label{line:step1-eg}
      \State Apply \texttt{Pairwise Classify} (Algorithm \ref{alg:initial-block}) on input $G, V_1, a, b$ to obtain a labeling $\widehat{\sigma}$ of $V_1$. \label{line:step2-eg}
      \For{$i=2,\cdots, 2n/\log n$} \label{line:step3-eg}
      \State Apply \texttt{Propagate} (Algorithm \ref{alg:propagation}) on input $G, V_{i-1}, V_i$ to determine the labeling $\widehat{\sigma}$ on $V_i$. 
      \EndFor \label{line:step3-end-eg}
      \vspace{5pt}
      \State{{\bf Phase II:}}
      \For{$u\in V$}  
      \State Apply \texttt{Refine} (Algorithm \ref{alg:refine}) on input $G, \widehat{\sigma}, u$ to obtain $\widetilde{\sigma}(u)$.
      \EndFor
      \end{algorithmic}
  \end{breakablealgorithm}
  \footnotetext[1]{The number of blocks is $\ceil{2n/\log n}$ if $2n/\log n$ is not an integer.}
  
  \begin{breakablealgorithm}
      \caption{\texttt{Pairwise Classify}} 
      \label{alg:initial-block}
      \begin{algorithmic}[1]
      \Require{ Graph $G = (V,E)$, vertex set $S \subset V$, parameters $a, b \in [0,1]$ with $a \neq b$.}
      \State Choose an arbitrary vertex $u_0 \in S$, and set $\widehat{\sigma}(u_0) = 1$.
      \For{$u \in S \setminus\{u_0\}$}
      \If{$\left|\{v \in S \setminus \{u, u_0\}: \{u_0, v\}, \{u,v\} \in E\} \right| > (a+b)^2(|S| -2)/4$} \label{line:threshold}
      \State Set $\widehat{\sigma}(u) = 1$.
      \Else
          \State Set $\widehat{\sigma}(u) = -1$.
      \EndIf
      \EndFor
      \end{algorithmic}    
  \end{breakablealgorithm}
  
  \begin{breakablealgorithm}
      \caption{\texttt{Propagate}} \label{alg:propagation}
      \begin{algorithmic}[1]
      \Require{ Graph $G = (V,E)$, mutually visible sets of vertices $S, S' \subset V$ with $S \cap S' = \emptyset$, where $S$ is labeled according to $\widehat{\sigma}$.}
      \If{$|\{v \in S \colon \widehat{\sigma}(v) = 1\}| \geq |\{v \in S \colon \widehat{\sigma}(v) = -1\}|$}\label{line:step3-case1}
      \For{$u \in S'$}
      \If{$a > b \text{  and  }
      d_{1}^+(u, \widehat{\sigma}, S) \ge (a+b)\cdot|\{v \in S \colon \widehat{\sigma}(v) = 1\}|/2$}
      \State Set $\widehat{\sigma}(u) = 1$.
      \ElsIf{$a < b \text{  and  }
      d_{1}^+(u, \widehat{\sigma}, S) < (a+b)\cdot|\{v \in S \colon \widehat{\sigma}(v) = 1\}|/2$}
      \State Set $\widehat{\sigma}(u) = 1$.
      \Else
      \State Set $\widehat{\sigma}(u) = -1$.
      \EndIf
      \EndFor \label{line:step3-case1-end}
      \Else \label{line:step3-case2}
      \For{$u \in S'$}
      \If{$a > b \text{  and  }
      d_{-1}^+(u, \widehat{\sigma}, S) \ge (a+b)\cdot|\{v \in S \colon \widehat{\sigma}(v) = -1\}|/2$}
      \State Set $\widehat{\sigma}(u) = -1$.
      \ElsIf{$a < b \text{  and  }
      d_{-1}^+(u, \widehat{\sigma}, S) < (a+b)\cdot|\{v \in S \colon \widehat{\sigma}(v) = -1\}|/2$}
      \State Set $\widehat{\sigma}(u) = -1$.
      \Else
      \State Set $\widehat{\sigma}(u) = 1$.
      \EndIf
      \EndFor
      \EndIf\label{line:step3-case2-end}
      \end{algorithmic}
  \end{breakablealgorithm}

  \begin{breakablealgorithm}
      \caption{\texttt{Refine}} \label{alg:refine}
      \begin{algorithmic}[1]
      \Require{ Graph $G \sim \text{GSBM}(\lambda, n,a,b,d)$, vertex $u \in V$, labeling $\widehat{\sigma}: V \to \{-1,0,1\}$.}
      \Ensure{ An estimated labeling $\widetilde{\sigma}(u) \in \{-1,1\}$.}  
      \State{
      Set
      $
      \widetilde{\sigma}(u) = \text{sign}\Big[\log\big(\dfrac{a}{b} \big)\big(d_{1}^+(u,\widehat{\sigma}) - d_{-1}^+(u,\widehat{\sigma}) \big)+\log\big(\dfrac{1-a}{1-b} \big)\big(d_{1}^-(u,\widehat{\sigma}) - d_{-1}^-(u,\widehat{\sigma}) \big) \Big].
      $}
      \end{algorithmic}
  \end{breakablealgorithm}
  
  In Phase II, we refine the almost-exact labeling $\widehat{\sigma}$ obtained from Phase I. Our refinement procedure mimics the so-called \emph{genie-aided} estimator \cite{Abbe2017}, which labels a vertex $u$ knowing the labels of all other vertices (i.e., $\{\sigma_0(v)\colon v\in V \setminus\{u\}\}$).
  The degree profile relative to the ground-truth labeling, $d(u,\sigma_0)$, is random and depends on realizations of node locations and edges in $G$ and community assignment $\sigma_0$.
  We use $D\in \mathbb{R}^4$ to denote the vector representing the four random variables in $d(u,\sigma_0)$.
  Then  $D$  is characterized by a multi-type Poisson distribution such that conditioned on $\{\sigma_0(u)=1\}$, $D\sim\text{Poisson}(\lambda \nu_d \log n [a, 1-a,b,1-b ]/2 )$ and conditioned on $\{\sigma_0(u)=-1\}$,  $D\sim\text{Poisson}(\lambda \nu_d \log n [b,1-b,a, 1-a]/2)$.
  Given a realization $D=d(u,\sigma_0)$, we pick the most likely hypothesis to minimize the error probability; that is,
  \begin{align}
  \sigma_{\textsf{genie}}(u) &= \argmax_{s\in\{1, -1\}} \pr(D=d(u,\sigma_0)\given \sigma_0(u)=s) \nonumber\\
  &=\text{sign}\Big[\log\big(\frac{a}{b} \big)\big(d_{1}^+(u,\sigma_0) - d_{-1}^+(u,\sigma_0) \big)+\log\big(\frac{1-a}{1-b} \big)\big(d_{1}^-(u,\sigma_0) - d_{-1}^-(u,\sigma_0) \big)\Big]. \label{eq:genie}
  \end{align}
  For convenience, let 
  \begin{equation}
  \tau(u, \sigma) = \log\big(\frac{a}{b}\big)\big(d_{1}^+(u,\sigma) - d_{-1}^+(u,\sigma)\big)+\log\big(\frac{1-a}{1-b} \big)\big(d_{1}^-(u,\sigma) - d_{-1}^-(u,\sigma)\big). \label{eq:tau}   
  \end{equation} 
  In short, we have $\sigma_{\textsf{genie}}(u) = \text{sign}(\tau(u, \sigma_0))$. The genie-aided estimator motivates the \texttt{Refine} subroutine (Algorithm \ref{alg:refine}) in Phase II that assigns $\widetilde\sigma(u) = \text{sign}(\tau(u,\widehat\sigma))$ for any $u\in V$. 
  Since $\widehat{\sigma}$ makes few errors compared with $\sigma_0$, for any $u\in V$, its degree profile $d(u, \widehat\sigma)$ is close to $d(u,\sigma_0)$. Thus, $d(u,\widehat{\sigma})$ is well-approximated by the aforementioned multi-type Poisson distribution.  
  
  \vskip6pt
  \paragraph*{Modified algorithm for general $\lambda > 1$.} 
  If $1<\lambda<2$, partitioning the interval into blocks of length $\log n/2$, as done in Line \ref{line:step1-eg} of Algorithm \ref{alg:almost-exact-eg}, fails. This is because each of the $2n/\log n$ blocks is independently empty with probability $e^{-\lambda\log n/2} = n^{-\lambda/2}$ and $-\lambda/2>-1$, leading to a high probability of encountering empty blocks, and thus a failure of the propagation scheme. 
  To address this, we instead adopt smaller blocks of length $\chi\log n$, where $\chi < (1-1/\lambda)/2$, for any $\lambda>1$. We only attempt to label blocks with sufficiently many vertices, according to the following definition. For the rest of the paper, let $V(B)\subset V$ denote the set of vertices in a subregion $B\subset\cS_{d,n}$. 
  \begin{definition}[Occupied block]
      Given any $\delta>0$, a block $B \subset \cS_{d,n}$ is $\delta$-occupied if $|V(B)| > \delta\log n$. Otherwise, $B$ is $\delta$-unoccupied.
  \end{definition}
  
  We will show that for sufficiently small $\delta > 0$, all but a negligible fraction of blocks are $\delta$-occupied. As a result, achieving almost-exact recovery in Phase I only requires labeling the vertices within the occupied blocks. To ensure successful propagation, we introduce a notion of visibility. Two blocks $B_i, B_j \in \cS_{d,n}$ are \emph{mutually visible}, defined as $B_i \sim B_j$, if 
  \[
      \sup_{x \in B_i, y \in B_j} \Vert x - y \Vert \leq (\log n)^{1/d}.\]
  Thus, if $B_i \sim B_j$, then any pair of vertices $u \in B_i$ and $v \in B_j$ are at a distance at most $(\log n)^{1/d}$ of each other. In particular, if $B_j$ is labeled and $B_i \sim B_j$, then we can propagate labels to $B_i$.
  
  Similar to the case of $\lambda > 2$, we propagate labels from left to right. Despite the presence of unoccupied blocks, we establish that if $\lambda>1$ and $\chi$ is chosen as above, each block $B_i$ following the initial $B_1$ has a corresponding block $B_j$ ($j<i$) to its left that is occupied and satisfies $B_i\sim B_j$. We thus modify Lines \ref{line:step3-eg}-\ref{line:step3-end-eg} so that a given block $B_i$ is labeled by one of the visible, occupied blocks to its left (Figure \ref{fig:propagate}). The modification is formalized in the general algorithm (Algorithm \ref{alg:almost-exact}) given below.
  
  \subsection{Exact recovery for general \texorpdfstring{$d$}{}.} 
  The propagation scheme becomes more intricate for $d \geq 2$. For general $d$, we divide the region $\cS_{d,n}$ into hypercubes\footnote{For $d=1,2,3$, the hypercubes represent line segments, squares and cubes respectively.} with volume parametrized as $\chi \log n$. The underlying intuition for successful propagation stems from the condition $\lambda \nu_d>1$. This condition ensures that the graph formed by connecting all pairs of mutually visible vertices is connected with high probability, a necessary condition for exact recovery. Moreover, the condition ensures that every vertex has $\Omega(\log n)$ vertices within its visibility radius of $(\log n)^{1/d}$. It turns out that the condition $\lambda \nu_d>1$ also ensures that blocks of volume $\chi \log n$ for $\chi > 0$ sufficiently small satisfy the same connectivity properties.
  
  To propagate the labels, we need a schedule to visit all occupied blocks. However, the existence of unoccupied blocks precludes the use of a predefined schedule, such as a lexicographic order scan. Instead, we employ a data-dependent schedule. The schedule is determined by the set of occupied blocks, which in turn is determined in Step \ref{sample:step-1} of Definition \ref{def:gsbm}. Crucially, the schedule is thus independent of the community labels and edges, conditioned on the number of vertices in each block. We first introduce an auxiliary graph $H = (V^{\dagger}, E^{\dagger})$, which records the connectivity relation among occupied blocks.
  \begin{definition}[Visibility graph]
  Consider a Poisson point process $V\subset\cS_{d,n}$, the $(\chi\log n)$-block partition of $\cS_{d,n}$, $\{B_i\}_{i=1}^{n/(\chi\log n)}$, corresponding vertex sets $\{V_i\}_{i=1}^{n/(\chi\log n)}$, and a constant $\delta>0$. The $(\chi, \delta)$-visibility graph is denoted by $H = (V^{\dagger}, E^{\dagger})$, where the vertex set $V^{\dagger} = \{i \in [n/(\chi\log n)] : |V_i| \geq \delta \log n\}$ consists of all $\delta$-occupied blocks and the edge set is given by $E^{\dagger}=\{\{i,j\}\colon i,j \in V^{\dagger}, B_i \sim B_j\}$.
  \end{definition}
  We adopt the standard connectivity definition on the visibility graph. Lemma \ref{lemma:connectivity} shows that the visibility graph of the Poisson point process underlying the GSBM is connected with high probability. Based on this connectivity property, we establish a propagation schedule as follows. We construct a spanning tree of the visibility graph and designate a root block as the initial block. We specify an ordering of $V^{\dagger} = \{i_1, i_2, \dots \}$ according to a tree traversal (e.g., breadth-first search). Labels are propagated according to this ordering, thus labeling vertex sets $V_{i_1}, V_{i_2}, \cdots$ (see Figure \ref{fig:propagate}). Letting $p(i)$ denote the parent of vertex $i \in V^{\dagger}$ according to the rooted tree, we label $V_{i_j}$ using $V_{p(i_j)}$ as reference.
  Importantly, the visibility graph and thus the propagation schedule is determined only by the locations of vertices, independent of the labels and edges between mutually visible blocks.

  \begin{figure}[htbp]
  \centering
  \includegraphics[width=.91\linewidth]{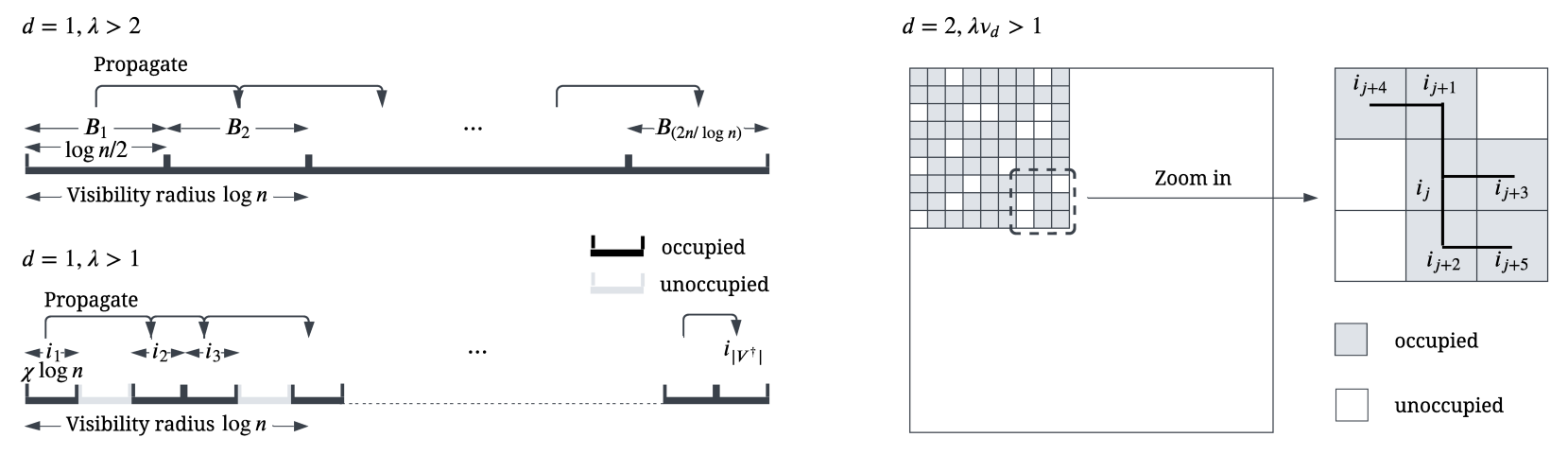}
    \caption{Propagation schedule for $d=1$ and $d=2$.}
    \label{fig:propagate}
  \end{figure}

  \begin{breakablealgorithm}
      \caption{Exact recovery for the GSBM} \label{alg:almost-exact}
      \begin{algorithmic}[1]
      \Require $G \sim \text{GSBM}(\lambda, n,a,b,d)$.
      \Ensure{ An estimated community labeling $\widetilde{\sigma}: V \to \{-1,1\}$.}
      \vspace{5pt}
      \State{{\bf Phase I:}} 
      \State Take small enough $\chi,\delta>0$, satisfying the conditions to be specified in \eqref{eq:chi-fomula} and \eqref{eq:delta-formula} respectively. 
      \State Partition the region $\cS_{d,n}$ into $n/(\chi\log n)$ 
   blocks of volume $\chi\log n$ each. Let $B_i$ be the $i$th block and $V_i$ be the set of vertices in $B_i$ for $i\in [n/(\chi\log n)]$. \label{line:partition}
   \State Form the associated visibility graph $H = (V^{\dagger}, E^{\dagger})$. \label{line:visibility-graph}
      \If{$H$ is disconnected}\label{line:H}
      \State Return \texttt{FAIL}.
      \EndIf
      \State Find a rooted spanning tree of $H$, ordering $V^{\dagger} = \{i_1, i_2, \cdots\}$ in breadth-first order. \label{line:tree-order}
  \State{Apply \texttt{Pairwise Classify} (Algorithm \ref{alg:initial-block}) on input $G, V_{i_1}, a, b$ to obtain a labeling $\widehat{\sigma}$ of $V_{i_1}$.}\label{line:step2}  
      \For{$j=2,\cdots, |V^\dagger|$} \label{line:step3}
      \State{Apply \texttt{Propagate} (Algorithm \ref{alg:propagation}) on input $G, V_{p(i_j)}, V_{i_j}$ to determine the labeling $\widehat{\sigma}$ on $V_i$. }
      \EndFor \label{line:step3-end}
      \For{$u \in V \setminus \left(\cup_{i \in V^{\dagger}} V_i\right)$}
      \State Set $\widehat{\sigma}(u) = 0$.
      \EndFor
      \vspace{5pt}
      \State{{\bf Phase II:}}
      \For{$u\in V$}  
      \State Apply \texttt{Refine} (Algorithm \ref{alg:refine}) on input $G, \widehat{\sigma}, u$ to determine $\widetilde{\sigma}(u)$.
      \EndFor
      \end{algorithmic}
  \end{breakablealgorithm}
  
  Algorithm \ref{alg:almost-exact} presents our algorithm for the general case. We partition the region $\cS_{d,n}$ into blocks with volume $\chi \log n$, for a suitably chosen $\chi > 0$. A threshold level of occupancy $\delta>0$ is specified. The value of $\chi$ is carefully chosen to ensure that the visibility graph $H$ is connected with high probability in Line \ref{line:H}. In Line \ref{line:step2}, we label an initial $\delta$-occupied block, corresponding to the root of $H$, using the \texttt{Pairwise Classify} subroutine. In Lines \ref{line:step3}-\ref{line:step3-end}, we label the occupied blocks in the tree order determined in Line \ref{line:tree-order}, using the \texttt{Propagate} subroutine. Those vertices appearing in unoccupied blocks are assigned a label of $0$. At the end of Phase I, we obtain a first-stage labeling $\widehat{\sigma}\colon V \to \{-1,0,1\}$, such that with high probability, all occupied blocks are labeled with few mistakes. Finally, Phase II refines the almost-exact labeling $\widehat{\sigma}$ to an exact one $\widetilde{\sigma}$ 
  according to Algorithm \ref{alg:refine}.
  
  To analyze the runtime, note that the number of edges (input size) is $\Theta(n \log n)$ with high probability.
  The visibility graph $H = (V^{\dagger}, E^{\dagger})$ can be formed in $O(n/\log n)$ time, since $|V^{\dagger}| = O(n/\log n)$ and each vertex has at most $\Theta(1)$ possible neighbors. If $H$ is connected, a spanning tree can be found in $O(|E^{\dagger}| \log(|E^{\dagger}|))$ time using Kruskal's algorithm, and $|E^{\dagger}| = O(n/\log n)$.
  The subsequent \texttt{Pairwise Classify} subroutine goes over all edges of the vertices in $V_1$ to count the common neighbors, with a runtime of $O(\log^2 n)$. 
  Next, the \texttt{Propagation} subroutine requires counting edges and non-edges from any given vertex in an occupied block to the vertices in its reference block, yielding a runtime of $O(n \log n)$.
  Finally, \texttt{Refine} runs in $O(n \log n)$ time, since each visible neighborhood contains $O(\log n)$ vertices.
  We conclude that Algorithm \ref{alg:almost-exact} runs in $O(n \log n)$ time, which is linear in the number of edges.

  \subsection{Proof outline.}
  We outline the analysis of Algorithm \ref{alg:almost-exact}. We begin with Phase I. Our goal is to show that in addition to achieving almost exact recovery stated in Theorem \ref{theorem:almost-exact-recovery}, Phase I also satisfies an error dispersion property. Let $\cN(u) = \{v\in V, \|u-v\|\le (\log n)^{1/d}\}$ for a vertex $u$. Namely, for any $\eta>0$, we can take suitable $\chi, \delta >0$ so that with high probability, every vertex has at most $\eta \log n$ incorrectly classified vertices in its local neighborhood $ \cN(u)$. Theorem \ref{thm:phase1-summary} will present the formal results.
  
  \vskip6pt
  \paragraph*{Phase I: Connectivity of the visibility graph.}
  We first establish that the block division specified in Algorithm \ref{alg:almost-exact} ensures that the resulting visibility graph $H = (V^{\dagger}, E^{\dagger})$ is connected.
  Elementary analysis shows that any fixed subregion of $\mathbb{R}^d$ with volume $\nu \log n$ contains $\Omega(\log n)$ vertices with probability $1-o(n^{-1})$, whenever $\nu > 1/{\lambda}$.
  A union bound over all vertices then implies that all vertices' neighborhoods have $\Omega(\log n)$ vertices. In the special case of $d=1$, the left neighborhood of a given vertex has volume $\log n$. The observation with $\nu = 1$ implies that when $\lambda > 1$, the left neighborhood of every vertex has $\Omega(\log n)$ points. In fact, we can make a stronger claim: if the block lengths are chosen to be sufficiently small (according to \eqref{eq:chi-fomula}), then we can ensure that for a given vertex $v \in V_i$, there are $\Omega(\log n)$ vertices among $\{V_j : B_j \sim B_i, j \neq i\}$. In turn, by an appropriate choice of $\delta$ (according to \eqref{eq:delta-formula}), for a given block $B_i$, there is at least one $\delta$-occupied, visible block to its left. Hence, the visibility graph is connected, as shown in Proposition \ref{lem:visibility-d1-small-lambda}.
  
  However, the analysis becomes more intricate when $d\geq 2$. In particular, while a lexicographic order propagation schedule succeeds for $d = 1$, it fails for $d \geq 2$. For example, when $d = 2$, we cannot say that every vertex has $\Omega(\log n)$ vertices in the top left quadrant of its neighborhood, since the volume of the quadrant is only $\nu_d \log n/4$.
  We therefore establish connectivity of $H$ using the fact that if $H$ is disconnected, then $H$ must contain an isolated connected component. The key idea is that if there is an isolated connected component in $H$, then the corresponding occupied blocks in $\mathbb{R}^d$ must be surrounded by sufficiently many unoccupied blocks. However, as Lemma \ref{lem:unoccupied-cluster} shows, there cannot be too many adjacent unoccupied blocks, which prevents the existence of isolated connected components. As a result, the visibility graph is connected, as shown in Lemma \ref{lemma:connectivity}.
  
  \vskip6pt
  \paragraph*{Phase I: Labeling the initial block.}
  We show that the \texttt{Pairwise Classify} (Line \ref{line:step2}) subroutine ensures the successful labeling for $V_{i_1}$. Since we only need to determine community labels up to a global flip, we are free to set $\widehat{\sigma}(u_0) = 1$ for an arbitrary $u_0 \in V_{i_1}$.
  For any $u\in V_{i_1}\setminus\{u_0\}$, where $|V_{i_1}|=m_1$, Lemma \ref{lem:N_u,v-Binomial} shows that the number of common neighbors of $u$ and $u_0$ follows a binomial distribution; in particular, Bin$(m_1-2, (a^2 + b^2)/2)$ if $\sigma_0(u)=\sigma_0(u_0)$ and Bin$(m_1-2, ab)$ otherwise. We thus threshold the number of common neighbors in order to classify $u$ relative to $u_0$. Lemma \ref{lem:step2} bounds the probability of misclassifying a given vertex $u \in V_{i_1}\setminus\{u_0\}$, using Hoeffding's inequality. A union bound then implies that all vertices are correctly classified with high probability. 
  
  \vskip6pt
  \paragraph*{Phase I: Propagating labels among occupied blocks.} 
  We show that the \texttt{Propagate} subroutine ensures that $\widehat{\sigma}$ makes at most $M$ mistakes in each occupied block, where $M$ is a suitable constant. Our analysis reduces to bounding the probability that for a given $i \in V^{\dagger}$, the estimator $\widehat{\sigma}$ makes more than $M$ mistakes on $V_i$, conditioned on making no more than $M$ mistakes on $V_{p(i)}$. In order to analyze the probability that a given vertex $v \in V_i$ is misclassified, we condition on the \emph{label configuration} of $V_{p(i)}$, meaning the number of vertices labeled $s$ according to $\sigma_0(\cdot)$ and $t$ according to $\sigma_0(u_0) \widehat{\sigma}(\cdot)$, for $s,t \in \{-1, +1\}$. We find a uniform upper bound on the probability of misclassifying an individual vertex $v \in V_i$ when applying the thresholding test given in Algorithm \ref{alg:propagation}, over all label configurations of $V_{p(i)}$ with at most $M$ mistakes. To bound the total number of mistakes in $V_i$, observe that the labels of all vertices in $V_i$ are decided based on disjoint subsets of edges between $V_i$ and $V_{p(i)}$. Therefore, conditioned on the label configuration of $V_{p(i)}$, the number of mistakes in $V_i$ can be stochastically dominated by a binomial random variable. It follows by elementary analysis that the number of mistakes in $V_i$ is at most $M$ with probability $1-o(n^{-1})$, as long as $M$ is a suitably large constant.
  
  \vskip6pt
  \paragraph*{Phase II: Refining the labels.} 
  Our final step is to refine the initial labeling $\widehat{\sigma}$ from Phase I into a final labeling $\widetilde{\sigma}$. Unfortunately, conditioning on a successful labeling $\widehat{\sigma}$ destroys the independence of edges, making it difficult to bound the error probability of $\widetilde{\sigma}$. This issue can be remedied using a technique called \emph{graph splitting}, used in the two-round procedure of \cite{Abbe2015}. Graph splitting is a procedure to form two graphs, $G_1$ and $G_2$, from the original input graph. A given edge in $G$ is independently assigned to $G_1$ with probability $p$, and $G_2$ with probability $1-p$, for $p$ chosen so that almost exact recovery can be achieved on $G_1$, while exact recovery can be achieved on $G_2$. Since the two graphs are nearly independent, conditioning on the success of almost exact recovery in $G_1$ essentially maintains the independence of edges in $G_2$.
  
  While we believe that our Phase I algorithm, along with graph splitting, would achieve the information-theoretic threshold in the GSBM, we instead directly analyze the robustness of Poisson testing. Specifically, we bound the error probability of labeling a given vertex $v \in V$ with respect to the worst-case labeling over all labelings that differ from $\sigma_0$ on at most $\eta \log n$ vertices in the neighborhood of $v$. Since $\widehat\sigma$ makes at most $\eta \log n$ errors with probability $1-o(1/n)$ (Theorem \ref{thm:phase1-summary}), we immediately obtain a bound on the error probability of $\widetilde{\sigma}(v)$.
  
  The proof in Section \ref{sec:proof-phase-II} bounds the worst-case error probability. We define $x = \lambda \nu_d [a, 1-a,b,1-b ]/2$ and $y = \lambda \nu_d [b, 1-b,a,1-a ]/2$, so that $D\given \{\sigma_0(u)=1\}\sim\text{Poisson}(x)$ and $D\given \{\sigma_0(u)=-1\}\sim\text{Poisson}(y)$. The condition $\lambda \nu_d(1 - \sqrt{ab} - \sqrt{(1-a)(1-b)}) > 1$ in Theorem \ref{theorem:exact-recovery} is equivalent to $D_+(x \| y) > 1$, where $D_+(x \| y)$ is the Chernoff--Hellinger divergence of $x$ and $y$ \cite{Abbe2015}. To provide intuition for bounding the error probability at a given vertex $u\in V$, consider 
  the genie-aided estimator $\sigma_{\textsf{genie}}(u)$, and assume $\sigma_0(u)=1$ without loss of generality. Recalling the definition of $\tau$ \eqref{eq:tau}, the estimator $\sigma_{\textsf{genie}}(u)$ makes a mistake when $\tau(u,\sigma_0)\le0$. It can be shown that this occurs with probability at most $n^{-D_+(x\|y)}$. Viewing the worst-case labeling $\sigma$ differing from $\sigma_0$ on at most $\eta\log n$ vertices as a perturbation of $\sigma_0$, we show that $\tau(u,\sigma)\le0$ implies $\tau(u,\sigma_0)\le\rho\eta\log n$ for a certain constant $\rho$. Similarly, the probability of such a mistake is at most $n^{-D_+(x\|y)+\rho\eta/2}$. 
  Thus, for small $\eta>0$, the condition $D_+(x \| y) > 1$ and a union bound over all vertices yields an error probability of $o(1)$.
  
  \section{Phase I: Proof of almost exact recovery}\label{sec:proof-phase-I}
  In this section, we prove Theorem \ref{theorem:almost-exact-recovery}. We begin by defining sufficiently small constants $\chi$ and $\delta$ used in Algorithm \ref{alg:almost-exact}. We define $\chi$ to satisfy the following condition, relying on $\lambda$ and $d$:
  \#\label{eq:chi-fomula}
  \nu_d \big(1 - 3\sqrt{d}\chi^{1/d}/2  \big)^d \ge (\nu_d + 1/{\lambda} )/2 \text{ and } 0<\chi <[(\mathds{1}_{d=1} + \nu_d\cdot\mathds{1}_{d\ge2})- 1/{\lambda} ]/2.
  \#
  The first condition is satisfiable since $\lim_{\chi \to 0} \nu_d \big(1 - 3\sqrt{d}\chi^{1/d}/2  \big)^d = \nu_d$
  and we have $\nu_d>(\nu_d + 1/{\lambda} )/2$ when $\lambda \nu_d>1$. The second one is also satisfiable since $\mathds{1}_{d=1} + \nu_d\cdot\mathds{1}_{d\ge2} = 1 > 1/\lambda$ if $d = 1$ and otherwise $\mathds{1}_{d=1} + \nu_d\cdot\mathds{1}_{d\ge2} = \nu_d > 1/\lambda$, under the conditions of Theorems \ref{theorem:exact-recovery} and \ref{theorem:almost-exact-recovery}.
  Associated with the choice of $\chi$, there is a constant $\delta'(\text{or }\widetilde\delta\text{ for }d\ge 2)>0$ such that for any block $B_i$, its visible blocks $\bigcup_{j\in V}\{V_j\colon B_j \sim B_i\}$ contain at least $\delta'\log n$ (or $\widetilde\delta\log n$) vertices with probability $1 - o(n^{-1})$. We define $R_d =1- \sqrt{d}\chi^{1/d}/2$. The first condition in \eqref{eq:chi-fomula} implies that $\sqrt{d}\chi^{1/d}/2<1/3$ and thus $R_d>0$. 
  With specific values of $\delta'$ and $\widetilde\delta$ to be determined in Proposition \ref{lem:visibility-d1-small-lambda} and Lemma \ref{lem:fN_i-size}, respectively, we define $\delta$ such that
  \#\label{eq:delta-formula}
  0<\delta<(\delta'\chi)\cdot \mathds{1}_{d=1} + [\widetilde{\delta}\chi/(\nu_dR^d)]\cdot\mathds{1}_{d\ge2}.
  \#
  Propositions \ref{lem:visibility-d1-small-lambda} and \ref{lemma:connectivity} will present the connectivity properties of $\delta$-occupied blocks of volume $\chi \log n$, for $\chi$ and $\delta$ satisfying the conditions in \eqref{eq:chi-fomula} and \eqref{eq:delta-formula}, respectively.
  
  We now record some preliminaries (see \cite{boucheron2013concentration}).
  \begin{lemma}[Chernoff bound, Poisson]\label{lem:Chernoff-poisson} Let $X\sim\text{Poisson}(\mu)$ with $\mu>0$. For any $t>0$,
  \$
  \pr(X\ge \mu + t) 
  \le \exp\Big(-\frac{t^2}{2(\mu+t)}\Big).
  \$
  For any $0<t<\mu$, we have 
  \$
  \pr(X\le \mu - t) \le \exp\Big(-(\mu - t) \log \Big(1 - \frac{t}{\mu} \Big)-t\Big).
  \$
  \end{lemma}
  
  \begin{lemma}[Hoeffding's inequality]\label{lem:Hoeffding-bounded} 
          Let $X_1,\cdots, X_n$ be independent bounded random variables with values $X_i\in[0,1]$ for all $1\le i\le n$. Let $X=\sum_{i=1}^n X_i$ and $\mu = \mathbb{E}[X]$. Then for any $t\ge0$, it holds that 
          \$
          \pr(X \ge \mu+ t) \le \exp(-2t^2/n), \quad
          \pr(X \le \mu -t) \le \exp(-2t^2/n).
          \$
  \end{lemma}
  
  \begin{lemma}[Chernoff upper bound]\label{lem:Chernoff-binomial} 
  Let $X_1,\cdots, X_n$ be independent Bernoulli random variables. Let $X=\sum_{i=1}^n X_i$ and $\mu = \E(X)$. Then for any $t>0$, we have
      \$
      \pr(X\ge (1+t)\mu) \le \Big(\frac{e^t}{(1+t)^{(1+t)}} \Big)^{\mu}.
      \$
  \end{lemma} 

We also define a homogeneous Poisson point process used to generate locations as described in Definition \ref{def:gsbm}.
  \begin{definition}[\cite{kingman1992poisson}]\label{def:PPP}
    A homogeneous Poisson point process with intensity $\lambda$ on $S \subseteq \R^{d}$ is a random countable set $\Phi := \{v_1, v_2,\cdots \} \subset S$ such that
    \begin{enumerate}
        \item For any bounded Borel set $B\subset \R^d$, the count $N_\Phi(B) := |\Phi\cap B| = |\{i\in \N\colon v_i\in B\}|$ has a Poisson distribution with mean $\lambda\text{vol}(B)$, where $\text{vol}(B)$ is the measure (volume) of $B$.
        \item For any $k\in \N$ and any disjoint Borel sets $B_1,\cdots, B_k \subset \R^d$, the random variables $N_\Phi(B_1),$ $\cdots,$ $N_\Phi(B_k)$ are mutually independent.
    \end{enumerate}
\end{definition}
In the GSBM, the set of locations $V=\{v_1,v_2,\cdots\}$ are generated by a homogeneous Poisson point process with intensity $\lambda$ on $\cS_{n,d}$. The established properties guarantee that $|V|$ follows $\text{Poisson}(\lambda n)$. Moreover, conditioned on $|V|$, the locations $\{v_i\}_{i\in[|V|]}$ are independently and uniformly distributed in $\cS_{n,d}$. This gives a simple construction of a Poisson point process as follows:
\begin{enumerate}
    \item Sample $N_V \sim \text{Poisson}(\lambda n)$;
    \item Sample $v_1, \cdots, v_{N_V}$ independently and uniformly in the region $\cS_{n,d}$.
\end{enumerate}
This procedure ensures that the resulting set $\{v_1, \cdots, v_{N_V}\}$ constitutes a Poisson point process as desired.

  \subsection{Connectivity of the visibility graph.}
  
  In this subsection, we establish the connectivity of the visibility graph $H = (V^{\dagger}, E^{\dagger})$ from Line \ref{line:visibility-graph} of Algorithm \ref{alg:almost-exact}. The following lemma shows that regions of appropriate volume have $\Omega(\log n)$ vertices with high probability.
  \begin{lemma}\label{lemma:volume-points}
      For any fixed subset $B\subset\cS_{d,n}$ with a volume $\nu\log n$ such that $\lambda \nu>1$, there exist constants $0 < \gamma < \lambda \nu$ and $\epsilon>0$ such that \$\pr(|V(B)| > \gamma \log n) \geq 1- n^{-1-\epsilon}.\$ 
  \end{lemma}
  \begin{proof}
      For a subset $B$ with $\text{vol}(B)=\nu\log n$, we have $|V(B)|\sim \text{Poisson}(\lambda \nu\log n)$. To show the lower bound, we define a function $g:(0, \lambda \nu]\to\R$ as $g(x) = x(\log x - \log (\lambda \nu)) +\lambda \nu-x$. 
      It is easy to check that $g$ is continuous and decreases on $(0, \lambda \nu]$ with $\lim_{x\to0}g(x)=\lambda \nu$ and $g(\lambda \nu)=0$. When $\lambda \nu>1$, it holds that $\lim_{x\to0}g(x)=\lambda \nu>(1+\lambda \nu)/2$ and thus there exists a constant $\gamma \in (0, \lambda \nu)$ such that $g(\gamma) > (1+\lambda \nu)/2$. Thus, the Chernoff bound in Lemma \ref{lem:Chernoff-poisson} yields that 
      \$
      \pr(|V(B)|\le \gamma \log n) \le 
      \exp\Big(-[\gamma(\log \gamma - \log (\lambda \nu)) +\lambda \nu-\gamma]\log n\Big) = n^{-g(\gamma)} \leq n^{-(1+\lambda \nu)/2}.
      \$
      Taking $\epsilon = (\lambda \nu - 1)/2 >0$ concludes the proof.
  \end{proof}
  
  \subsubsection{The simple case when \texorpdfstring{\(d=1\)}{} \text{and} \texorpdfstring{\(\lambda >1\)}{}.}
  We start with the simple case when $d=1$.
  
  \vskip6pt
  \paragraph*{An example when $\lambda >2$.} We first study an example when $d=1$ and $\lambda>2$. If $\lambda>2$ and $\vol(B_i) = \log n/2$, we have $\lambda\vol(B_i)/\log n>1$, and thus Lemma \ref{lemma:volume-points} ensures the existence of positive constants $\gamma$ and $\epsilon$ such that $\pr(|V_i| > \gamma \log n) \geq 1 - n^{-1-\epsilon}$ for all $i\in[2n/\log n]$. Thus, the union bound gives that 
  \#\label{eq:connected-H-lambda2}
  \pr\Big(\bigcap_{i=1}^{2n/\log n}\big\{|V_i|> \gamma \log n\big\}\Big) = 1 - \pr\Big(\bigcup_{i=1}^{2n/\log n}\big\{|V_i|\le \gamma \log n\big\}\Big) \geq 1 - \frac{2n}{\log n}\cdot n^{-1-\epsilon} = 1-o(1).
  \#
  Since all blocks are $\gamma$-occupied, the $(1/2, \gamma)$-visibility graph $H = (V^{\dagger}, E^{\dagger})$ is trivially connected.

  \vskip6pt
  \paragraph*{General case when $\lambda >1$.} For small density $\lambda$, we partition the interval into small blocks and establish the existence of visible occupied blocks on the left side of each block.
  \begin{proposition}\label{lem:visibility-d1-small-lambda}
  If $d=1$ and $\lambda >1$, with $0<\chi <(1-1/\lambda)/2$, we consider the blocks $\{B_i\}_{i=1}^{n/(\chi\log n)}$ obtained from Line \ref{line:partition} in Algorithm \ref{alg:almost-exact}. Then there exists a constant $\delta'>0$ such that for any $0<\delta<\delta'\chi$, it holds that 
      \$
      \pr\Big(\bigcap_{i=1}^{n/(\chi\log n)}\big\{\exists j \colon j < i, B_j \sim B_i, \text{ and } B_j \text{ is $\delta$-occupied}\big\}\Big) = 1 - o(1).
      \$
  It follows that the $(\chi, \delta)$-visibility graph is connected with high probability.
  \end{proposition}
  \begin{proof}
      For any $i\in[n/(\chi\log n)]$, we define $ U_i = \bigcup_{j\colon j < i, B_j \sim B_i}{B}_j$ as the union of visible blocks on the left-hand side of ${B}_i$. We have $\text{vol}( U_i)=(\floor{1/\chi}-1)\chi\log n\ge (1 -2\chi)\log n$ and $\lambda\text{vol}( U_i)/\log n\ge \lambda(1-2\chi)>1$ when $\lambda>1$ and $\chi<(1-1/\lambda)/2$. Thus, Lemma \ref{lemma:volume-points} ensures the existence of positive constants $\delta'$ and $\epsilon$ such that $\pr( |\bigcup_{j\colon j < i, B_i \sim B_j}V_j| \le \delta' \log n) \leq n^{-1-\epsilon}$. We note that $|\{j\colon j <i, B_j\sim B_i\}| \le (\ceil{1/\chi} - 1) \le 1/\chi$. Thus, we take $0<\delta < \delta'\chi$ and obtain that
      \$
  \pr\Big(\bigcap_{j\colon j< i, B_j\sim B_i}\big\{|V_j|\le \delta\log n \big\}\Big) & \le  \pr\Big(\big|\bigcup_{j\colon j< i, B_j\sim B_i}V_j\big|\le \delta\log n /\chi \Big) \\
  &\le \pr\Big(\big|\bigcup_{j\colon j < i, B_j\sim B_i}V_j\big|\le \delta'\log n \Big) \leq n^{-1-\epsilon}.
  \$
  Therefore, the union bound over all $i\in[n/(\chi\log n)]$ gives
\$
  &\pr\Big(\bigcap_{i=1}^{n/(\chi\log n)}\big\{\exists j \colon j < i, B_j \sim B_i, \text{ and } B_j \text{ is $\delta$-occupied}\big\}\Big) \\
  & \quad = 1 - \pr\Big(\bigcup_{i=1}^{n/(\chi\log n)}\bigcap_{j\colon j < i, B_j\sim B_i}\big\{|V_j|\le \delta\log n \big\}\Big) \\
  &\quad \ge 1 - \frac{n}{\chi\log n}\cdot n^{-1-\epsilon} = 1 - o(1). \qedhere
\$
  \end{proof}
  
  \subsubsection{General case when \texorpdfstring{\(d\ge 2\)}{} \text{and} \texorpdfstring{\(\lambda \nu_d>1\)}{}.}
  We now study general cases. 
  We first show that for any block $B$, the set of surrounding visible blocks $\{B'\colon B \sim B', B' \neq B\}$ contains $\Omega(\log n)$ vertices. For any block $B_i\subset\cS_{d,n}$ with $\text{vol}(B_i)=\chi\log n$, the length of its longest diagonal is given by $\sqrt{d}(\chi\log n)^{1/d}$. Recall the definition of $R_d = 1- \sqrt{d}\chi^{1/d}/2$, and let $C_i$ be the ball of radius $R_d (\log n)^{1/d}$ centered at the center of $B_i$. Observe that
  \[\sup_{x \in B_i, y \in C_i} \Vert x - y \Vert = \frac{1}{2}\sqrt{d}(\chi\log n)^{1/d} + R_d (\log n)^{1/d} = (\log n)^{1/d}.\]
  It follows that if $B_j \subseteq C_i$, then $B_i \sim B_j$. Also, $C_i$ contains all blocks $B_j \sim B_i$ (see Figure \ref{fig:geometry}).
  We define 
  \[U_i = \bigcup_{j\colon j\neq i, B_j\sim B_i} B_j =\bigcup_{j \neq i \colon B_j\subset C_i} B_j\]
  as the union of all visible blocks to $B_i$, excluding $B_i$ itself. Observe that as $\chi \to 0$, the volume of the blue region approaches the volume of $C_i$. The following lemma quantifies this observation, showing that our conditions on $\chi$ guarantee that $U_i$ (and any set with the same volume as $U_i$) will contain sufficiently many vertices.

  \begin{figure}[htbp]
\centering\includegraphics[width=0.7\textwidth]{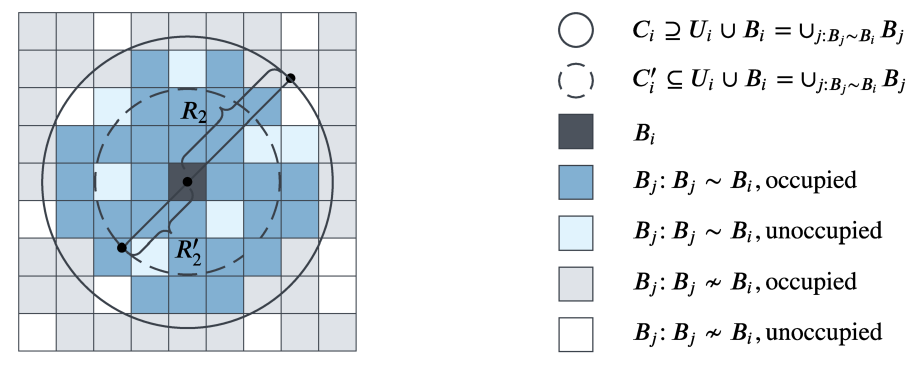}
      \caption{Geometry around block $B_i$, showing a portion of the region $\cS_{2,n}$. The set $U_i$ is comprised of dark and light blue blocks.}
      \label{fig:geometry}
  \end{figure}

  \begin{lemma}\label{lem:fN_i-size}
  If $\chi$ satisfies the condition in \eqref{eq:chi-fomula} and $\lambda \nu_d>1$, there exist positive constants $\widetilde\delta$ and $\epsilon$, depending on $\lambda$ and $d$, such that for any subset $S\in\cS_{d,n}$ with $\text{vol}(S) = \text{vol}( U_i)$, we have
  \$
  \pr\big(|V(S)|>\widetilde\delta\log n\big) \geq 1- n^{-1-\epsilon}.
  \$ 
  \end{lemma}
  \begin{proof}
  We first evaluate the volume of $U_i\subset C_i$.\footnote{This is similar to the Gauss circle problem \citep{ivic2006lattice}.} We define $R_d' = R_d-\sqrt{d}\chi^{1/d}$ and $C_i^{\prime}$ as the ball centered at the center of $B_i$ with a radius $R_d'(\log n)^{1/d}$. The condition in \eqref{eq:chi-fomula} implies that $3\sqrt{d}\chi^{1/d}/2<1$ and thus $R_d'>0$. Based on geometric observations, we note that $C_i^{\prime} \subset U_i\cup B_i\subset C_i$. It follows that $\text{vol}( U_i\cup B_i) \ge \text{vol}(C_i^{\prime}) = \nu_d(R_d')^d\log n$, and thus $\text{vol}( U_i) \ge (\nu_d(R_d')^d-\chi)\log n$.
  
  We now show that when $\lambda \nu_d>1$, the conditions in \eqref{eq:chi-fomula} imply $\lambda (\nu_d (R_d')^d - \chi)>1$ by observing the following relations:
  \begin{align*}
  \nu_d (R_d')^d-\chi &= \nu_d \big(1 - 3\sqrt{d}\chi^{1/d}/2  \big)^d - \chi\\
  &\geq (\nu_d + 1/{\lambda} )/2 - \chi\\
  &\geq 1/\lambda.
  \end{align*}
  In summary, we have shown that $\text{vol}(S) = \text{vol}( U_i)\ge (\nu_d(R_d')^d-\chi)\log n$ and $\lambda (\nu_d (R_d')^d - \chi)>1$. Thus, Lemma \ref{lemma:volume-points} ensures the existence of positive constants $\widetilde\delta$ and $\epsilon$ such that $\pr(|V(S)|>\widetilde\delta\log n) > 1- n^{-1-\epsilon}$.
  \end{proof}
  
  Henceforth, we use the term ``occupied block'' to refer to $\delta$-occupied blocks, as well as ``unoccupied block'', with the constant threshold $\delta=\delta(\lambda, d)$ defined in \eqref{eq:delta-formula} in the rest of the section. We define $K=|\{j\colon B_j\subset U_i\}|$ as the number of blocks in $U_i$, a constant relying on $\lambda$ and $d$. We note that $K\le\nu_d(R_d)^d/\chi - 1 < \widetilde{\delta}/\delta$ since $U_i\cup B_i\subset C_i$. The key observation in establishing connectivity is that there cannot be a large \emph{cluster} of unoccupied blocks.
  \begin{definition}[Cluster of blocks]
      Two blocks are adjacent if they share an edge or a corner. We say that a set of blocks $\mathcal{B}$ is a cluster if for every $B, B' \in \mathcal{B}$, there is a path of blocks of the form $(B = B_{j_1}, B_{j_2}, \dots, B_{j_m} = B')$, where $B_{j_k} \in \mathcal{B}$ for $k \in [m]$ and $B_{j_k}, B_{j_{k+1}}$ are adjacent.
  \end{definition}
  
  The following lemma shows that all clusters of unoccupied blocks have fewer than $K$ blocks, with high probability. This also implies that $U_i$ contains at least one occupied block for each $i$.
  \begin{lemma}\label{lem:unoccupied-cluster}
  Suppose $d\ge2$ and $\lambda \nu_d >1$. Let
  $Y$ be the size of the largest cluster of unoccupied blocks produced in Line \ref{line:partition} in Algorithm \ref{alg:almost-exact}. Then $\pr(Y < K) = 1-o(1)$.
  \end{lemma}
  \begin{proof}
  We first bound the probability that all $K$ blocks in any given set are unoccupied. For any set of $K$ blocks $\{B_{j_k}\}_{k=1}^K$, we have 
  \#\label{eq:pr-all-unoccupied}
  \pr\Big(\bigcap_{k=1}^K\big\{|V_{j_k}|\le\delta\log n\big\}\Big) & \le \pr\Big(\big|\bigcup_{k=1}^K V_{j_k}\big|\le\delta K \log n\Big) \notag\\
  & \le \pr\Big(\big|\bigcup_{k=1}^K V_{j_k}\big|<\widetilde\delta \log n\Big) \notag\\
  & \le n^{-1-\epsilon},
  \#
  where the second inequality holds due to $K<\widetilde\delta/\delta$ and the last inequality follows from Lemma \ref{lem:fN_i-size} and the fact that $\vol(\bigcup_{k=1}^K B_{j_k})=\vol(U_i)$.
  
  Let $Z$ be the number of unoccupied block clusters with a size of $K$. Then we have $\pr(Y\ge K) = \pr(Z \ge 1)$. Let $\fS$ be the set of all possible shapes of clusters of blocks with a size of $K$. Clearly, $|\fS|$ is a constant depending on $K$ and $d$. For any $s\in\fS$, $i\in [n/(\chi\log n)]$, and $j\in[K]$, we define $\cZ_{s, i,j}$ as the event that there is a cluster of unoccupied blocks, characterized by shape $s$ with block $B_i$ occupying the $j$th position. Due to \eqref{eq:pr-all-unoccupied}, we have $\pr(\cZ_{s, i, j})\le n^{-1-\epsilon}$. Thus, the union bound gives
  \$
  \pr(Y\ge K) &= \pr(Z \ge 1) = \pr\Big(\bigcup_{s\in\fS, i\in[n/(\chi\log n)], j\in[K]}\cZ_{s, i,j}\Big) \\
  &\le |\fS| \cdot \frac{n}{\chi\log n}\cdot  K \cdot n^{-1-\epsilon} =  o(1). \qedhere
  \$
  \end{proof}
   
  Finally, we establish the connectivity of the visibility graph.
  \begin{proposition}\label{lemma:connectivity}
  Suppose that $d \ge 2$ and $\lambda \nu_d > 1$. Let $V \subset \cS_{d,n}$ be a Poisson point process on $\cS_{d,n}$ with intensity $\lambda$. Then for $\chi$ and $\delta$ given in \eqref{eq:chi-fomula} and \eqref{eq:delta-formula}, respectively, the $(\chi, \delta)$-visibility graph $H$ on $V$ is connected with probability $1-o(1)$.
  \end{proposition}
  \begin{proof}
      For a visibility graph $H = (V^{\dagger},E^{\dagger})$, we say that $S \subset V^{\dagger}$ is a \emph{connected component} if the subgraph of $H$ induced on $S$ is connected.
      Let $\cE$ be the event that $H$ contains an isolated connected component. Formally, $\cE$ is the event that there exists $S\subset V^{\dagger}$ \footnote{The notation $\subset$ denotes a strict subset.} such that (1) $S \neq \emptyset$ and $S\neq V^\dagger$; (2) $S$ is a connected component; (3) for all $i \in S, j \not \in S$ we have $\{i,j\} \not \in E^{\dagger}$. Observe that $\{H \text{ is disconnected}\} = \cE$.
  
      For any $S \neq \emptyset$ and $S\subset V^\dagger$ to be an isolated connected component, it must be completely surrounded by a cluster of unoccupied blocks. In other words, all blocks in the cluster $(\bigcup_{i\in S}U_i)\setminus(\bigcup_{i\in S}B_i)$ must be unoccupied. We 
   next show that for any isolated, connected component $S$, we have $|\{j\colon B_j\subset (\bigcup_{i\in S}U_i)\setminus(\bigcup_{i\in S}B_i)\}| \ge K$; that is, the number of unoccupied blocks visible to an isolated connected component is at least $K$.
  
      We prove the claim by induction on $|S|$. In fact, we prove it for $S$ that is isolated, but not necessarily connected.
      The claim holds true whenever $|S| = 1$ by the definition of $K$. Suppose that the claim holds for every isolated component with $k$ blocks. Consider an isolated component $S$, with $|S| = k +1$. Let $F = (\bigcup_{i \in S} B_i) \bigcup (\bigcup_{i \in S} U_i)$ be the collective ``footprint'' of all elements of $S$ along with the surrounding unoccupied blocks. For each $j \in S$, let $F_j = (\bigcup_{i \in S, i \neq j} B_i) \bigcup (\bigcup_{i \in S, i \neq j} U_i)$ be the footprint of all blocks in $S$ excluding $j$. Let $G_{j}$ be the graph formed from $G$ by removing all vertices from $V_{j}$, thus rendering $V_{j}$ unoccupied. Observe that there must exist some $j^{\star} \in S$ such that $F_{j^{\star}} \neq F$ and $F_{j^{\star}} \subset F$, as the regions $\{B_i \cup U_i\}_{i \in S}$ are translations of each other. Since $S \setminus \{j^{\star}\}$ is an isolated component in $G_{j^{\star}}$, the inductive hypothesis implies that $S \setminus \{j^{\star}\}$ has at least $K$ surrounding unoccupied blocks in $G_{j^{\star}}$. Comparing $G_{j^{\star}}$ to $G$, there are two cases (see Figure \ref{fig:isolated-component} for examples in $\cS_{2,n}$). \emph{Case I.} In the first case, $F \setminus F_{j^{\star}}$ contains at least one unoccupied block. In that case, the inclusion of $V_{j^{\star}}$ changes one block from unoccupied to occupied, and increases the number of surrounding unoccupied blocks by at least one. Thus, $S$ contains at least $K$ surrounding unoccupied blocks. \emph{Case II.} In the second case, $F \setminus F_{j^{\star}} $ contains only occupied blocks. Since there are $k+1$ total occupied blocks in $F$ and $k$ of them are in $F_{j^\star}$, we have $F \setminus F_{j^{\star}} = B_{j^{\star}}$, so that $B_{j^{\star}} \cap F_{j^{\star}} = \emptyset$. In this case, the set of $K$ surrounding unoccupied blocks in $F_{j^\star}$ remains unoccupied in $F$.
      \begin{figure}[htbp]
  \centering
  \includegraphics[width=.85\linewidth]{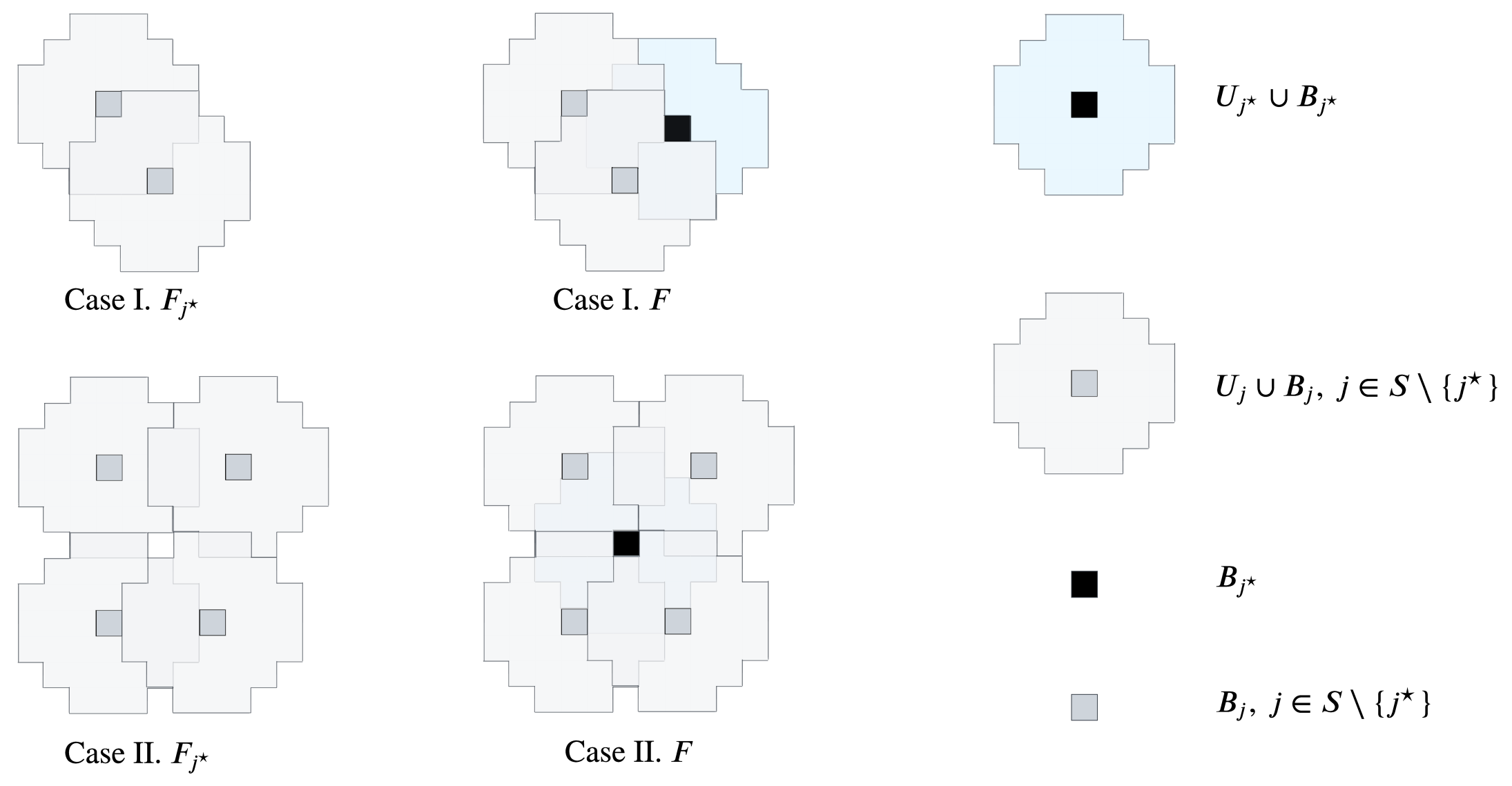}
    \caption{Possible isolated components in $\cS_{2,n}$ for Proposition \ref{lemma:connectivity}.}
    \label{fig:isolated-component}
  \end{figure}
  
      Thus, $\mathcal{E}$ implies $\{Y \geq K\}$. The result follows from Lemma \ref{lem:unoccupied-cluster}.
  \end{proof}

  In summary, Propositions \ref{lem:visibility-d1-small-lambda} and \ref{lemma:connectivity} establish the connectivity of visibility graphs for cases when $d=1$ and $\lambda>1$, or $d\ge 2$ and $\lambda \nu_d>1$, ensuring successful label propagation in the algorithm. For convenience, let $\cH = \{H \text{ is connected} \}.$ We conclude that $\pr(\cH) = 1- o(1)$.

  \subsection{Labeling the initial block.}
  We now prove that the \texttt{Pairwise Classify} subroutine (Line \ref{line:step2} of Algorithm \ref{alg:almost-exact}) ensures, with high probability, the correct labeling of all vertices in the initial block $V_{i_1}$. Let $N_{u_0,u} = |\{v\in V_{i_1}\colon \{v,u_0\}\in E,\{v,u\}\in E\}|$ be the number of common neighbors of $u_0$ and $u$ within $V_{i_1}$. 
  \begin{lemma}\label{lem:N_u,v-Binomial}
      For any vertex $u\in V_{i_1}\setminus\{u_0\}$, it holds that 
      \begin{enumerate}
          \item Conditioned on $\sigma_0(u) = \sigma_0(u_0)$ and $|V_{i_1}| = m_{i_1}$, we have $N_{u_0, u} \sim\mathrm{Bin}(m_{i_1}-2, (a^2+b^2)/2)$.
          \item  Conditioned on $\sigma_0(u) \neq \sigma_0(u_0)$ and $|V_{i_1}| = m_{i_1}$, we have $N_{u_0, u} \sim\mathrm{Bin}(m_{i_1}-2, ab)$.
      \end{enumerate}
  \end{lemma}
  \begin{proof}
  We first consider the case when $\sigma_0(u) = \sigma_0(u_0)$. For any vertex $v\in V_{i_1}\setminus \{u,u_0\}$, we have 
      \$
      & \pr\big((v, u)\in E, (v, u_0)\in E \given \sigma_0(u) = \sigma_0(u_0) \big) \\
      & \quad = \pr\big((v,u)\in E, (v, u_0)\in E\given \sigma_0(v)=\sigma_0(u), \sigma_0(u) = \sigma_0(u_0) \big)\pr\big( \sigma_0(v)=\sigma_0(u) \big) \\
      & \qquad + \pr\big((v, u)\in E, (v, u_0)\in E\given \sigma_0(v)\neq\sigma_0(u), \sigma_0(u) = \sigma_0(u_0) \big)\pr\big( \sigma_0(v)\neq\sigma_0(u) \big)  \\
      & \quad = (a^2 + b^2)/2.
      \$
      The first statement follows from mutual independence of the events $\{(v,u), (v,u_0) \in E\}$ over $v \in V_{i_1} \setminus \{u, u_0\}$, conditioned on $|V_{i_1}| = m_{i_1}$.
      
      Similarly, if $\sigma_0(u) \neq \sigma_0(u_0)$, for any $v\in V_{i_1}\setminus \{u,u_0\}$, we have 
      \$
      & \pr\big((v, u)\in E, (v, u_0)\in E \given \sigma_0(u) \neq \sigma_0(u_0) \big)  \\
      & \quad = \pr\big((v,u)\in E, (v, u_0)\in E\given \sigma_0(v)=\sigma_0(u), \sigma_0(u) \neq \sigma_0(u_0) \big)\pr\big( \sigma_0(v)=\sigma_0(u) \big) \\
      & \qquad + \pr\big((v, u)\in E, (v, u_0)\in E\given \sigma_0(v)\neq\sigma_0(u), \sigma_0(u) \neq\sigma_0(u_0) \big)\pr\big( \sigma_0(v)\neq\sigma_0(u) \big)  \\
      & \quad = ab,
      \$ 
      implying the second statement.
  \end{proof}
  
  The following lemma will be used to bound the misclassification probability of $u \in V_{i_1} \setminus \{u_0\}$ using the thresholding rule given in Algorithm \ref{alg:initial-block}, Line \ref{line:threshold}. Let $\cT_{u_0,u}=\{N_{u_0,u} > (a + b)^2(|V_{i_1}|-2)/4\}$. We define constants $\eta_1=\exp[(a-b)^4/4]$ and $c_1=\delta(a-b)^4/8$. 
  \begin{lemma}\label{lem:step2}
  For any vertex $u\in V_{i_1}\setminus\{u_0\}$ and any $m_{i_1}\ge\delta\log n$, we have
  \$
  \max\Big\{\pr\big(\cT_{u_0,u}^c \biggiven \sigma_0(u) = \sigma_0(u_0), |V_{i_1}|=m_{i_1}\big), \  \mathbb{P}\big(\cT_{u_0,u} \biggiven \sigma_0(u) \neq \sigma_0(u_0), |V_{i_1}|=m_{i_1}\big) \Big\} \le \eta_1n^{-c_1}.
  \$
  \end{lemma}
  \begin{proof}
  Fix $m_{i_1} \geq \delta \log n$.
  Lemma \ref{lem:N_u,v-Binomial} along with Hoeffding's inequality (Lemma \ref{lem:Hoeffding-bounded}) gives that 
  \$
  & \pr \big(\cT_{u_0,u}^c \biggiven \sigma_0(u) = \sigma_0(u_0), |V_{i_1}| = m_{i_1} \big) \\
  & 
  = \pr\big(N_{u_0,u} - (a^2 + b^2)(m_{i_1}-2)/2 \le -(a-b)^2(m_{i_1}-2)/4 \biggiven \sigma_0(u) = \sigma_0(u_0), |V_{i_1}| = m_{i_1}\big) \\
  & \le \exp\big(-(a-b)^4(m_{i_1}-2)/8 \big)\\
  &\leq \exp(-(a-b)^4(\delta \log n-2)/8) = \eta_1n^{-c_1}.
  \$
  Similarly,
  \$
  & \pr\big(\cT_{u_0,u} \biggiven \sigma_0(u) \neq \sigma_0(u_0), |V_{i_1}| = m_{i_1}\big) \\
  & \quad = \pr\big(N_{u_0,u} - ab(m_{i_1}-2) > (a-b)^2(m_{i_1}-2)/4 \biggiven \sigma_0(u) \neq \sigma_0(u_0), |V_{i_1}| = m_{i_1}\big) \\
  & \quad \le \exp\big(-(a-b)^4(m_{i_1}-2)/8 \big) \le \eta_1n^{-c_1}. \qedhere
  \$
  \end{proof}
  
  The following proposition ensures the high probability of correct labeling for all vertices in $V_{i_1}$.
  \begin{proposition}\label{prop:V1}
  Suppose that $a,b\in[0,1]$ with $a\neq b$. Then Line \ref{line:step2} of Algorithm \ref{alg:almost-exact} ensures that for any $\Delta>\delta$, 
      \$
      \pr \Big(\bigcap_{u \in V_{i_1}}\big\{\widehat{\sigma}(u) = \sigma_0(u_0)\sigma_0(u)\big\} \biggiven \delta \log n \le |V_{i_1}| \le \Delta\log n\Big) \geq 1- \eta_1\Delta n^{-c_1} \log n.
      \$
  \end{proposition}
  \begin{proof}
  For any $u\in V_{i_1}\setminus\{u_0\}$, when $m_{i_1}\ge\delta\log n$, Lemma \ref{lem:step2} implies
      \$
      & \pr\big(\widehat\sigma(u) \neq \sigma_0(u_0)\sigma_0(u) \biggiven |V_{i_1}| =m_{i_1}\big) \\
      & \quad = \pr\big(\widehat\sigma(u) = -1 \biggiven \sigma_0(u) = \sigma_0(u_0), |V_{i_1}| =m_{i_1}\big)\pr\big(\sigma_0(u) = \sigma_0(u_0)\big) \\
      & \qquad +  \pr\big(\widehat\sigma(u) = 1 \biggiven \sigma_0(u_0) \neq \sigma_0(u), |V_{i_1}| =m_{i_1}\big)\pr\big(\sigma_0(u_0) \neq \sigma_0(u)\big)  \\
      & \quad = \pr\big(\cT_{u,u_0}^c \biggiven \sigma_0(u) = \sigma_0(u_0), |V_{i_1}|=m_{i_1}\big)/2 + \mathbb{P}\big(\cT_{u, u_0} \biggiven \sigma_0(u) \neq \sigma_0(v), |V_{i_1}|=m_{i_1}\big)/2 \\
      & \quad \le \eta_1n^{-c_1}.
      \$
      Thus, for any $\delta\log n\le m_{i_1}\le\Delta\log n$, the union bound yields that
      \$
      \pr \Big(\bigcap_{u \in V_1}\big\{\widehat{\sigma}(u) = \sigma_0(u_0)\sigma_0(u)\big\} \biggiven |V_{i_1}| =m_{i_1} \Big) &= 1 - \pr\Big(\bigcup_{u \in B_1}\big\{\widehat{\sigma}(u) \neq \sigma_0(u_0)\sigma_0(u)\big\} \biggiven |V_{i_1}| =m_{i_1} \Big) \\
      &\ge 1-m_{i_1} \eta_1 n^{-c_1} \ge 1- \eta_1\Delta n^{-c_1} \log n.
      \$
      It follows that 
      \$
      \pr \Big(\bigcap_{u \in V_{i_1}}\big\{\widehat{\sigma}(u) = \sigma_0(u_0)\sigma_0(u)\big\} \biggiven \delta \log n \le |V_{i_1}| \le \Delta\log n\Big) \ge 1- \eta_1\Delta n^{-c_1} \log n. 
      \$ \qedhere
  \end{proof}
  
  \subsection{Propagating labels among occupied blocks.}
  We now demonstrate that the \texttt{Propagate} subroutine (Lines \ref{line:step3}-\ref{line:step3-end} of Algorithm \ref{alg:almost-exact}) ensures that all occupied blocks are classified with at most $M$ mistakes, for a suitable constant $M$.
  
  We introduce a vector $m=(m_1, \cdots, m_{(n/(\chi\log n))})\in \mathbb{Z}_{+}^{(n/(\chi\log n))}$ and define the event \[\cV(m)=\{|V_i| = m_i \text{ for } i\in[n/(\chi\log n)]\}.\] 
  Each $\cV(m)$ corresponds to a specific $(\chi,\delta)$-visibility graph $H$. Thus, conditioned on an event $\cV(m)$ that ensures the connectivity of $H$, the occupied block set $V^\dagger$ and the propagation ordering over $V^\dagger$ are uniquely determined. To simplify the analysis, we fix the vector $m$ in what follows, and condition on some event $\mathcal{V}(m) \subset \mathcal{H}$, recalling that $\mathcal{H} = \{H \text{ is connected}\}$. We write $\mathbb{P}_m(\cdot) = \mathbb{P}\left( \cdot \given \mathcal{V}(m)\right)$ as a reminder. Note that conditioned on $\mathcal{V}(m)$, the labels of vertices are independent, and the edges are independent conditioned on the vertex labels.

  We denote the \emph{configuration} of a block as a vector 
  $z = (z(1,1), z(1,-1), z(-1,-1), z(-1,1))\in Z_{+}^4$, where each entry represents the count of vertices labeled as $+1$ or $-1$ by $\sigma_0$ and $\widehat{\sigma}$. For $i\in V^\dagger$, the event $\cC_{i}(z)$ signifies that the occupied block $V_{i}$ possesses a configuration $z$ such that  
  \begin{align*}
  &|\{u\in V_{i}, \sigma_0(u) = \sigma_0(u_0), \widehat\sigma(u) = 1\}| = z(1,1)\\
  &|\{u\in V_{i}, \sigma_0(u) = \sigma_0(u_0), \widehat\sigma(u) = -1\}| = z(1,-1)\\
  &|\{u\in V_{i}, \sigma_0(u) \neq \sigma_0(u_0), \widehat\sigma(u) = -1\}| = z(-1,-1)\\
  &|\{u\in V_{i}, \sigma_0(u) \neq \sigma_0(u_0), \widehat\sigma(u) =1\}| = z(-1,1).
  \end{align*}
  Consider $i \in V^{\dagger} \setminus \{i_1\}$ and a configuration $z \in \mathbb{Z}_+^4$. The key observation is that because the labels $\{\widehat\sigma(u) : u \in V_i\}$ are determined using disjoint sets of edges, the labels $\{\widehat\sigma(u) : u \in V_i\}$ are independent conditioned on $\cC_{p(i)}$. Thus, the number of mistakes on $V_i$ can be dominated by a binomial random variable. To formalize this observation, we define constants $M=5/[(a-b)^2\delta]$, $c_2=(a-b)^2\delta/4$, and $\eta_2= \exp(2(a-b)^2M)$. Let $\cA_{i}$ be the event that $\widehat\sigma$ makes at most $M$ mistakes on $V_{i}$: 
  \$
  \cA_{i} = \big\{ |\{u \in V_{i} : \widehat{\sigma}(u) \neq \sigma_0(u_0) \sigma_0(u)\} | \leq M\big\}.
  \$
  The following lemma bounds the probability of misclassifying a given vertex using Algorithm \ref{alg:propagation}.
  \begin{lemma}\label{lem:degree-profile-condition}
  Suppose that $a, b \in [0,1]$ and $a \neq b$, and fix $i\in V^\dagger\setminus\{i_1\}$. Fix $z \in \mathbb{Z}_+^4$ such that $z(1,1) + z(1,-1) + z(-1,-1) + z(-1,1) = m_{p(i)}$ and $z(1,-1)  +z(-1,1) \leq M$ (so that $\cC_{p(i)}(z) \subset \cA_{p(i)}$). Then for any $u\in V_i$, we have \$
  \pr_m\big(\widehat\sigma(u) \neq \sigma_0(u_0)\sigma_0(u)\biggiven \cC_{p(i)}(z)\big) \le \eta_2n^{-c_2}.
  \$
  \end{lemma}
  
  \begin{proof} 
  We consider the case $a > b$. Let $\cJ_{+} = \{|\{v \in V_{p(i)} : \widehat\sigma(v) = 1 \}| \geq |\{v \in V_{p(i)} : \widehat\sigma(v) = -1\}|\}$. We first study the case when $\cJ_{+}$ holds. In this context, Lines \ref{line:step3-case1}-\ref{line:step3-case1-end} of Algorithm \ref{alg:propagation} are executed. Conditioned on any $\cC_{p(i)}(z)$, we have $|\{v \in V_{p(i)} \colon\widehat\sigma(v) = 1 \}| = z(1,1) + z(-1,1)$. Among these vertices $v\in V_{p(i)}$ with $\widehat\sigma(v)=1$, $z(1,1)$ vertices have ground truth label $\sigma_0(u_0)$ and $z(-1,1)$ of them have label $-\sigma_0(u_0)$. We now bound the probability of making a mistake, meaning that $\widehat\sigma(u)\neq\sigma_0(u_0)\sigma_0(u)$.
  
  If $\sigma_0(u) = \sigma_0(u_0)$, let $\{X_i\}_{i=1}^{z(1,1)}$ and $\{Y_i\}_{i=1}^{z(-1,1)}$ be independent random variables with $X_i\sim\text{Bernoulli}(a)$ and $Y_i\sim\text{Bernoulli}(b)$, and $Z=\sum_{i=1}^{z(1,1)}X_i + \sum_{i=1}^{z(-1,1)}Y_i$ with mean $\mu_Z=z(1,1)a+z(-1,1) b$. For any $u\in V_i$, we recall that $d_1^+(u, \widehat\sigma, V_{p(i)}) = |\{v \in V_{p(i)} : \widehat\sigma(v) = 1, \{u,v\} \in E\}|$ and observe that conditioned on $\{\sigma_0(u) = \sigma_0(u_0), \cC_{p(i)}(z)\big\}$, the degree profile $d_1^+(u, \widehat\sigma, V_{p(i)})$ has the same distribution as $Z$.
  Thus, Hoeffding's inequality yields
  \$
  &\pr_m\big(\widehat\sigma(u) \neq 1\biggiven \sigma_0(u) = \sigma_0(u_0), \cC_{p(i)}(z)\big) \\
  & \quad = \pr_m\Big(d_{1}^+(u, \widehat{\sigma}, V_{p(i)})< (a+b)|\{v\in V_{p(i)}: \widehat\sigma(v) = 1\}|/2 \Biggiven  \sigma_0(v) = \sigma_0(u_0), \cC_{p(i)}(z)\Big)  \\
  & \quad = \pr_m\big(Z < (a +b)(z(1,1) + z(-1,1))/2 \big) \\
  & \quad = \pr_m\big(Z-\mu_Z < -(a-b)(z(1,1) - z(-1,1))/2\big) \\
  & \quad \le \exp\big(- \frac{(a-b)^2(z(1,1) - z(-1,1))^2}{2(z(1,1) + z(-1,1))}\big).
  \$
  We recall that $\cJ_{+}$ implies $|\{v \in V_{p(i)} : \widehat\sigma(v) = 1 \}|\ge|V_{p(i)}|/2\ge\delta\log n/2$, and $z(1,-1) + z(-1,1) \leq M$. It follows that $z(1,1) + z(-1,1)\ge\delta\log n/2$ and $z(1,1) \geq \delta\log n/2 - M$.
  Thus,
  \begin{align}
  &\pr_m\big(\widehat\sigma(u) \neq 1\biggiven \sigma_0(u) = \sigma_0(u_0), \cC_{p(i)}(z)\big) \le \exp\big(- \frac{(a-b)^2(z(1,1) - M)^2}{2(z(1,1) + M)}\big) \nonumber\\
  & \quad \le \exp\big(-(a-b)^2(z(1,1)-3M)/2 \big)  \le \eta_2\exp\big(-(a-b)^2\delta \log n/4 \big) = \eta_2 n^{-c_2},  \label{eq:propagation-mistake+}
  \end{align}
  where the last two inequalities hold since $(z(1,1)-M)^2/ (z(1,1) + M) \ge z(1,1) -3M$ and $z(1,1) \ge \delta \log n/2 - M$. 
  
  Similarly, when $\sigma_0(u) \neq \sigma_0(u_0)$, let $\{X_i\}_{i=1}^{z(-1,1)}$ and $\{Y_i\}_{i=1}^{z(1,1)}$ be independent random variables with $X_i\sim\text{Bernoulli}(a)$ and $Y_i\sim\text{Bernoulli}(b)$, and $\widetilde Z=\sum_{i=1}^{z(1,1)}Y_i + \sum_{i=1}^{z(-1,1)}X_i$ with mean $\mu_{\widetilde Z}=z(1,1)b+z(-1,1) a$. For any $u\in V_{i}$, we observe that $d_1^+(u, \widehat\sigma, V_{p(i)})$ has the same distribution as $\widetilde Z$, conditioned on $\big\{\sigma_0(u) \neq \sigma_0(u_0), \cC_{p(i)}(z)\big\}$.
  By similar steps as the case $\sigma_0(u) = \sigma_0(u_0)$, we obtain
  \begin{align}
  &\pr_m\big(\widehat\sigma(u) \neq -1 \biggiven \sigma_0(v) \neq \sigma_0(u_0), \cC_{p(i)}(z)\big) \nonumber\\
  & \quad = \pr_m\Big(d_{1}^+(u, \widehat{\sigma}, V_{p(i)}) \ge (a+b)|\{v\in V_{p(i)}: \widehat\sigma(v) = 1\}|/2 \Biggiven  \sigma_0(v) \neq \sigma_0(u_0), \cC_{p(i)}(z)\Big) \nonumber\\
  & \quad = \pr\big(\widetilde Z \ge (a +b)(z(1,1) + z(-1,1))/2\big) \nonumber\\
  & \quad = \pr\big(\widetilde Z-\mu_{\widetilde Z} \ge (a-b)(z(1,1) - z(-1,1))/2\big) \nonumber\\
  &\quad \le \exp\big(- \frac{(a-b)^2(z(1,1) - z(-1,1))^2}{2(z(1,1) + z(-1,1))}\big)\nonumber\\
  &\quad \leq \eta_2n^{-c_2}. \label{eq:propagation-mistake-}
  \end{align}
  The bounds \eqref{eq:propagation-mistake+} and \eqref{eq:propagation-mistake-} together imply
  \[\pr_m\big(\widehat\sigma(u) \neq \sigma_0(u_0)\sigma_0(u)\biggiven \cC_{p(i)}(z) \big) \le \eta_2 n^{-c_2}.\]
  
  We can derive symmetric analysis for $z$ such that $\cJ_{+}^c$ holds, in which case Algorithm \ref{alg:propagation} executes Lines \ref{line:step3-case2}-\ref{line:step3-case2-end}. The proof is complete for the case $a > b$. The analysis for the case $b > a$ is similar.
  \end{proof}
  
  Before proceeding further and showing the success of the propagation, we state a lemma that, with high probability, all blocks contain $O(\log n)$ vertices. 
  
  \begin{lemma}\label{lem:B-upper-bound}
  For the blocks obtained from Line \ref{line:partition} in Algorithm \ref{alg:almost-exact}, there exists a constant $\Delta>0$ such that 
  \$
  \pr\big(\bigcap_{i=1}^{n/(\chi\log n)} \big\{|V_i|< \Delta\log n\big\}\big) = 1- o(1).
  \$
  \end{lemma}
  \begin{proof}
  For a block ${B}_i$ with $\text{vol}({B}_i)=\chi\log n$, we have $|V_i|\sim \text{Poisson}(\lambda \chi\log n)$. 
  Thus, the Chernoff bound in Lemma \ref{lem:Chernoff-poisson} implies that, for $\Delta > (\lambda \chi + 1+\sqrt{2\lambda \chi+1})$, we have 
  \$ \pr(|{V}_i|\ge \Delta\log n) \le \exp\Big(-\frac{(\Delta-\lambda \chi)^2\log n}{2\Delta}\Big) = n^{-\frac{(\Delta-\lambda \chi)^2}{2\Delta}} < n^{-1}, 
  \$ 
  where the last inequality holds by straightforward calculation. Thus, the union bound gives that
  \$
  \pr\Big(\bigcap_{i=1}^{n/(\chi\log n)} \big\{|V_i|< \Delta\log n\big\}\Big) = 1 - \pr\Big(\bigcup_{i=1}^{n/(\chi\log n)} \big\{|V_i| \ge \Delta\log n\big\}\Big) > 1 - \frac{n}{\chi\log n}\cdot n^{-1} = 1- o(1).
  \$
  \end{proof}
  
  For $\Delta>0$ given by Lemma \ref{lem:B-upper-bound}, we define $\cI$ as follows and have $\pr(\cI)=1-o(1)$. 
  \$\cI = \bigcap_{i=1}^{n/(\chi\log n)} \{|V_i|< \Delta\log n\}.\$ The following lemma concludes that Phase I makes few mistakes on occupied blocks during the propagation. 
  \begin{lemma}\label{prop:phase1-occupied} Let $G \sim \text{GSBM}(\lambda, n,a,b,d)$ with $\lambda \nu_d> 1$, $a, b \in [0,1]$, and $a \neq b$, and $\widehat{\sigma} : V \to \{-1, 0, 1\}$ be the output of Phase I in Algorithm \ref{alg:almost-exact} on input $G$. 
  Suppose $m$ is such that $\cV(m)\subset\cI\cap\cH$. Lines \ref{line:step3}-\ref{line:step3-end} of Algorithm \ref{alg:almost-exact} ensure that 
  \$\pr_m\Big(\bigcap_{i\in V^\dagger}\cA_{i}\Big) \geq \big(1- \eta_1\Delta n^{-c_1} \log n \big) \Big(1 - \frac{\eta_3 n^{-\frac{1}{8}}}{\chi \log n} \Big).\$
  \end{lemma}
  \begin{proof}
  Consider $i_j \in V^{\dagger}$ for $2 \leq j \leq |V^{\dagger}|$, and fix $z \in \mathbb{Z}_+^4$ such that 
  \begin{equation}
  z(1,1) + z(1,-1) + z(-1,-1) + z(-1,1) = m_{p(i_j)} \text{ and } z(1,-1)  +z(-1,1) \leq M. \label{eq:z-conditions}
  \end{equation}
  Observe that the events that $u \in V_{i_j}$ is mislabeled by $\widehat\sigma$ are mutually independent conditioned on $\cC_{p(i_j)}(z)$. Lemma \ref{lem:degree-profile-condition} shows that each individual vertex in $V_{i_j}$ is misclassified with probability at most $\eta_2n^{-c_2}$, conditioned on $\cC_{p(i_j)}(z)$. It follows that conditioned on $\cC_{p(i_j)}(z)$,
  \[|\{u \in V_{i_j} : \widehat{\sigma}(u) \neq \sigma_0(u_0) \sigma_0(u) \}| \stleq \text{Bin}\left(\Delta\log n, \eta_2n^{-c_2} \right) =: \xi.\]
  Let $\mu_{\xi} = \mathbb{E}[\xi] = \eta_2\Delta n^{-c_2}\log n$.
  Using the Chernoff bound (Lemma \ref{lem:Chernoff-binomial}), we obtain
  \begin{align}
  \pr_m\big(\cA_{i_j}^c \biggiven \cC_{p(i_j)}(z) \big) &= \pr_m\big(|\{u \in V_{i_j} : \widehat{\sigma}(u) \neq \sigma_0(u_0) \sigma_0(u)\} | > M \biggiven \cC_{p(i_j)}(z)\big) \nonumber \\
      & \le \pr(\xi > M)\nonumber \\
      & = \pr\big(\xi - \mu_\xi> (M/\mu_\xi- 1)\mu_\xi \big)\nonumber\\
  & \le e^{M-\mu_\xi}(\mu_\xi/M)^{M} \nonumber\\
  &\le (e\eta_2\Delta/M)^M(\log n)^M n^{-c_2M}\nonumber \\
  & \le \eta_3 n^{-9/8}.\label{eq:mistake-bound}    
  \end{align}
  The last inequality holds since $c_2M=5/4$ by definition and $(\log n)^M\le n^{1/8}$ for large enough $n$. Since $\cA_{i_j}$ is independent of $\{\cA_{i_k}: k < j, k \neq p(i_j)\}$ conditioned on $\cC_{p(i_j)}$, \eqref{eq:mistake-bound} implies 
  \begin{equation*}
  \pr_m\Big(\cA_{i_j}^c \biggiven \cC_{p(i_j)}(z), \bigcap_{k < j: i_k \neq p(i_j)}\cA_{i_k} \Big) \leq \eta_3 n^{-9/8}. 
  \end{equation*}
  Furthermore, since \eqref{eq:mistake-bound} is a uniform bound over all $z$ satisfying \eqref{eq:z-conditions}, it follows that
  \begin{equation*}
  \pr_m\Big(\cA_{i_j}^c \biggiven  \bigcap_{k < j}\cA_{i_k} \Big) \leq \eta_3 n^{-9/8}. 
  \end{equation*}
  
  Thus, combining Proposition \ref{prop:V1} with the preceding bound, we have
  \$
  \pr_m\Big(\bigcap_{i\in V^\dagger}\cA_{i}\Big) &= \pr_m\big( \cA_{i_1}\big)\cdot\prod_{j=2}^{|V^{\dagger}|}\pr_m\big(\cA_{i_j} \biggiven \cA_{i_{j-1}}, \cdots, \cA_{i_1}\big) \\
  & \ge \big(1- \eta_1\Delta n^{-c_1} \log n \big) \Big(1 - \eta_3 n^{-\frac{9}{8}} \Big)^{|V^{\dagger}|-1}\\
  &\geq \big(1- \eta_1\Delta n^{-c_1} \log n \big) \Big(1 - \eta_3 n^{-\frac{9}{8}} \Big)^{\frac{n}{\chi \log n}}\\
  &\geq \big(1- \eta_1\Delta n^{-c_1} \log n \big) \Big(1 - \frac{\eta_3 n^{-\frac{1}{8}}}{\chi \log n} \Big), 
  \$
  where we use the fact that there are ${n}/{\chi \log n}$ blocks in total along with Bernoulli's inequality.
  \end{proof}
  
  Combining the aforementioned results, we now prove the success of Phase I in Theorem \ref{thm:phase1-summary}. We highlight that since $\eta > 0$ is arbitrary, the following equation \eqref{eq:almost-exact-recovery} implies Theorem \ref{theorem:almost-exact-recovery}.
  \begin{theorem}\label{thm:phase1-summary}
  Given GSBM$(\lambda, n, a, b, d)$ with $a, b \in [0,1]$, $a \neq b$, and $d=1$ and $\lambda>1$, or $d\ge 2$ and  $\lambda \nu_d>1$. Fix any $\eta > 0$. Let $\kappa = \nu_d(1+\sqrt{d}\chi^{1/d})^d/\chi$. Let $\widehat\sigma$ be the labeling obtained from Phase I with $\chi>0$ satisfying \eqref{eq:chi-fomula} and $\delta>0$ satisfying \eqref{eq:delta-formula} and $\delta<\eta/\kappa$, respectively. 
  Then there exists a constant $M$ such that $\widehat{\sigma}$ makes at most $M$ mistakes on every occupied block, with high probability, 
  \begin{equation}
  \mathbb{P}\Big(\bigcap_{i \in V^{\dagger}} \big\{|\{v \in V_i: \widehat{\sigma}(v) \neq \sigma_0(u_0) \sigma_0(v)\}| \leq M\big\} \Big) = 1-o(1). \label{eq:occupied-block-mistakes}
  \end{equation}
  Moreover, it follows that 
  \begin{equation}
  \mathbb{P}\big(|\{v \in V : \widehat{\sigma}(v) \neq \sigma_0(u_0) \sigma_0(v)\} | \leq \eta n/(\chi\kappa) \big) = 1-o(1)   \label{eq:almost-exact-recovery} 
  \end{equation}
  and
  \begin{equation}
  \pr\Big(\bigcap_{u\in V}\big\{|v\in\cN(u)\colon \widehat\sigma(v)\neq\sigma_0(u_0)\sigma_0(v)|\le \eta\log n\big\}\Big) = 1-o(1). \label{eq:dispersed-errors} 
  \end{equation}
  \end{theorem}
  \begin{proof} Fixing any $\eta>0$, we consider $\chi>0$ satisfying \eqref{eq:chi-fomula} and $\delta>0$ satisfying \eqref{eq:delta-formula} and $\delta<\eta/\kappa$, respectively.
  Given any $m$ such that $\cV(m)\subset\cI\cap\cH$, for occupied blocks,
  Proposition \ref{prop:phase1-occupied} yields the existence of a constant $M>0$ such that \$\pr_m\Big(\bigcap_{i\in V^\dagger}\{ |\{v \in V_i : \widehat{\sigma}(v) \neq \sigma_0(u_0) \sigma_0(v)\} | \leq M \}\Big) \geq \big(1- \eta_1\Delta n^{-c_1} \log n \big) \Big(1 - \frac{\eta_3 n^{-\frac{1}{8}}}{\chi \log n} \Big).\$ 
  Since the above bound is uniform over all $m$ such that $\cV(m) \subset \cI \cap \mathcal{H}$, we have
  \$
  &\pr\Big(\bigcap_{i\in V^{\dagger}}\big\{|v\in V_i\colon \widehat\sigma(v)\neq\sigma_0(u_0)\sigma_0(v)|\le \delta\log n\big\}\Big) \\
  &\quad \ge \sum_{m\colon \cV(m) \subset \cI \cap \cH} \pr_m\Big(\bigcap_{i\in V^{\dagger}}\big\{|v\in V_i\colon \widehat\sigma(v)\neq\sigma_0(u_0)\sigma_0(v)|\le \delta\log n\big\}\Big) \cdot \pr\big(\cV(m)\big) \\
  &\quad \ge \big(1- \eta_1\Delta n^{-c_1} \log n \big) \Big(1 - \frac{\eta_3 n^{-\frac{1}{8}}}{\chi \log n} \Big)\cdot  \pr\big(\cI\cap\cH\big) = 1-o(1), 
  \$
  where the last step holds by Propositions \ref{lem:visibility-d1-small-lambda} and \ref{lemma:connectivity}, and Lemma \ref{lem:B-upper-bound}. Thus, we have proven \eqref{eq:occupied-block-mistakes}.
  
  Since $\delta \log n > M$ for $n$ large enough, it follows that
  \begin{equation}
  \pr\Big(\bigcap_{i \in [n/\chi\log n]}\{ |\{v \in V_i : \widehat{\sigma}(v) \neq \sigma_0(u_0) \sigma_0(v)\} | \leq \delta \log n \}\Big) = 1-o(1).   \label{eq:block-mistakes} 
  \end{equation}
  
  On the one hand, if $\widehat\sigma$ makes fewer than $\delta\log n$ mistakes on $V_i$ for all $i\in[n/(\chi\log n)]$, then $\widehat\sigma$ makes fewer than $\delta n/\chi \le \eta n/(\chi\kappa)$ mistakes in $\mathcal{S}_{d,n}$. Thus, \eqref{eq:almost-exact-recovery} follows from \eqref{eq:block-mistakes}. On the other hand, if $\widehat\sigma$ makes fewer than $\delta\log n$ mistakes on $V_i$ for all $i\in[n/(\chi\log n)]$, then there will be fewer than $\delta\kappa\log n\le \eta\log n$ mistakes in all vertices' neighborhood since each neighborhood $\cN(u)$ intersects at most $\kappa$ blocks. Thus, \eqref{eq:dispersed-errors} also follows from \eqref{eq:block-mistakes}.
  \end{proof}
  
  \section{Phase II: Proof of exact recovery}\label{sec:proof-phase-II}
  Before proving Theorem \ref{theorem:exact-recovery}, we first show a concentration bound. We define vectors in $\R^4$,
  \#\label{eq:def-x-y}
  x = \lambda \nu_d \log n [a, 1-a,b,1-b ]/2, \quad y = \lambda \nu_d \log n [b, 1-b, a, 1-a]/2,
  \# 
  and random variables $\widetilde D=[D_1^+, D_1^-, D_{-1}^+, D_{-1}^-] \sim \text{Poisson} (x)$, and $X$ as a linear function of $\widetilde D$,
  \#\label{eq:def_X}
  X = -\log\big(\frac{a}{b} \big)\big(D_{1}^+ - D_{-1}^+ \big)-\log\big(\frac{1-a}{1-b} \big)\big(D_{1}^- - D_{-1}^- \big).
  \# 
  For any $t\in[0,1]$, let $D_t(x\| y) = \sum_{i\in[4]}(tx_{i} + (1-t)y_{i} - x_{i}^ty_{i}^{1-t})$ be an $f$-divergence. Let $D_+(x\| y)= \max_{t\in[0,1]}D_t(x\| y) = \max_{t\in[0,1]}D_t(y\| x)$ be the Chernoff-Hellinger divergence, as introduced by \cite{Abbe2015}. In particular, when $x$ and $y$ are defined in \eqref{eq:def-x-y}, the maximum is achieved at $t = 1/2$ and we have $D_+(x\|y) = \lambda \nu_d(1-\sqrt{ab} - \sqrt{(1-a)(1-b)})\log n$.
  \begin{lemma}\label{lem:bound-X}
  For any constants $\rho>0$ and $\eta>0$, it holds for $X$ defined in \eqref{eq:def_X} that
  \$
  \pr\big(X\ge -\rho\eta\log n \big) \le n^{-\lambda \nu_d(1-\sqrt{ab} - \sqrt{(1-a)(1-b)})+\rho\eta/2}.
  \$
  \end{lemma}
  \begin{proof}
  We will apply the Chernoff bound on $X$. First, we compute its moment-generating function.
  For $\widetilde D=[D_1^+, D_1^-, D_{-1}^+, D_{-1}^-]=(D_i)_{i=1}^4\sim\text{Poisson}(x)$, the definition of $X$ in \eqref{eq:def_X} can be written as 
  \$
  X = -\sum_{i=1}^4 [D_i\log (x_i/y_i) - (x_i - y_i)].\$ 
  We recall that for $\xi\sim \text{Poisson}(\mu)$ and $s\in\R$, we have $\E[\exp(s\xi)]= \exp[\mu(e^s-1)]$. Thus, we have
  \$
  \E(e^{tX}) & = \E\Big[\exp\big(-t\sum_{i=1}^4 [D_i\log (x_i/y_i) - (x_i - y_i)]\big)\Big] \\
  & = \prod_{i=1}^4 \exp(t(x_i - y_i)) \cdot \E[\exp(t\log(y_i/x_i)D_i)] \\
  & = \prod_{i=1}^4 \exp(t(x_i - y_i) + x_i(e^{t\log(y_i/x_i)}-1)) \\
  & = \prod_{i=1}^4 \exp((t-1)x_i - ty_i + x_i^{1-t}y_i^t) \\
  & = \exp(-\sum_{i=1}^4((1-t)x_i + ty_i - x_i^{1-t}y_i^t)) = \exp(-D_t(y\|x)).
  \$
  Therefore, the Chernoff bound ensures that for any $t>0$, we have
  \$
  &\pr\big(X\ge -\rho\eta\log n \big) \le \frac{\E(e^{tX})}{e^{-t\rho\eta\log n}} = n^{t\rho\eta}\cdot\exp(-D_t(y\|x)).
  \$
  It follows that
  \$
  \pr\big(X\ge -\rho\eta\log n \big)
  &\le \inf_{t>0} \big\{n^{t\rho\eta}\cdot\exp(-D_t(y\|x))\big\} \\
  & \le n^{\rho\eta/2}\cdot\exp(-D_+(x\|y)) \\
  &= n^{-\lambda \nu_d(1-\sqrt{ab} - \sqrt{(1-a)(1-b)})+\rho\eta/2}. \qedhere 
  \$
  \end{proof}
  
  Now we present the proof of Theorem \ref{theorem:exact-recovery},  which ensures that Algorithm \ref{alg:almost-exact} achieves exact recovery.
  \begin{proof}[Proof of Theorem \ref{theorem:exact-recovery}]
  We first fix a constant $c>\lambda$ and let $\cE_0= \{|V| < c n\}$. Since $|V|\sim\text{Poisson}(\lambda n)$, the Chernoff bound in Lemma \ref{lem:Chernoff-poisson} gives that
  \$
  \pr(\cE_0^c) = \pr(|V|> c n) \le \exp\big(-\frac{(c-\lambda)^2n}{2c}\big) = o(1).
  \$
  For $\eta > 0$ to be determined, let $\cE_1$ be the event that $\widehat\sigma$ makes at most $\eta \log n$ mistakes in the neighborhood for all vertices (Phase I succeeds); that is,
  \$
  \cE_1=\bigcap_{u\in V}\big\{|v\in\cN(u)\colon \widehat\sigma(v)\neq\sigma_0(u_0)\sigma_0(v)|\le \eta\log n\big\}.
  \$
  Theorem \ref{thm:phase1-summary} ensures that $\pr(\cE_1)=1-o(1)$.
  Let $\cE_2'$ be the event that Algorithm \ref{alg:almost-exact} achieves exact recovery and $\cE_2$ be the event that all vertices are labeled correctly relative to $\sigma_0(u_0)$; that is, 
  \[
  \mathcal{E}_2' = \Big\{\bigcap_{u \in V}\{\widetilde{\sigma}(u) = \sigma_0(u)\}\Big\}\bigcup\Big\{\bigcap_{u \in V}\{\widetilde{\sigma}(u) = -\sigma_0(u)\}\Big\},\quad\mathcal{E}_2 = \bigcap_{u \in V}\{\widetilde{\sigma}(u) = \sigma_0(u_0) \sigma_0(u)\}.
  \]
  Then we have $\pr(\cE_2')\ge\pr(\cE_2)$. Since $\mathbb{P}(\mathcal{E}_0), \mathbb{P}(\mathcal{E}_1) = 1-o(1)$, it follows that
  \begin{equation}
  \pr(\cE_2^c) \leq \pr(\cE_2^c \cap \cE_1 \cap \cE_0) + \pr(\cE_1^c) + \pr(\cE_0^c) = \pr(\cE_2^c \cap \cE_1 \cap \cE_0) + o(1). \label{eq:phase-II-failure}
  \end{equation}
  
  Note that we analyze $\pr(\cE_2^c \cap \cE_1 \cap \cE_0)$ rather than $\pr(\cE_2^c \given \cE_1, \cE_0)$, in order to preserve the data distribution. Next, we would like to show that the probability of misclassifying a vertex $v$ is $o(1/n)$, and conclude that the probability of misclassifying \emph{any} vertex is $o(1)$. To formalize such an argument, sample $N \sim \text{Poisson}(\lambda n)$, and generate $\max\{N, cn\}$ points in the region $\mathcal{S}_{d,n}$ uniformly at random. Note that on the event $\mathcal{E}_0$, we have $\max\{N, cn\} = cn$. Label the points in order, and set $\widehat\sigma(u_0) = 1$. In this way, the first $N$ points form a Poisson point process with intensity $\lambda$. We can simulate Algorithm \ref{alg:almost-exact} on the first $N$ points. To bound the failure probability of Phase II, we can assume that any $v\in \{N+1, \dots, cn\}$ must also be classified (by thresholding $\tau(v,\sigma)$, computed only using edge/non-edge observations between $v$ and $u\in [N]$)).
  For $v \in [cn]$, let
  \[\mathcal{E}_2(v) = \{\widetilde{\sigma}(v) = \sigma_0(u_0) \sigma_0(v)\}.\]
  Then
  \[\mathcal{E}_2^c \cap \mathcal{E}_1 \cap \mathcal{E}_0 \subseteq \bigcup_{v=1}^{cn} \{\mathcal{E}_2(v)^c \cap \mathcal{E}_1 \cap \mathcal{E}_0\} \subseteq \bigcup_{v=1}^{cn} \{\mathcal{E}_2(v)^c \cap \mathcal{E}_1\},\]
  so that a union bound yields
  \begin{equation}\mathbb{P}\left(\mathcal{E}_2^c \cap \mathcal{E}_1 \cap \mathcal{E}_0 \right) \leq \sum_{v=1}^{cn} \mathbb{P}\left(\mathcal{E}_2(v)^c \cap \mathcal{E}_1 \right). \label{eq:union-bound}
  \end{equation}
  
  Fix $v \in [cn]$. In order to bound $\mathbb{P}\left(\mathcal{E}_2(v)^c \cap \mathcal{E}_1 \right)$, we classify $v$ according to running the \texttt{Refine} algorithm with respect to edge/non-edge observations between $v$ and $u \in [N]$. Analyzing $\mathcal{E}_2(v)^c \cap \mathcal{E}_1$ now reduces to analyzing robust Poisson testing.
  Let $W(v) = \{\sigma : \mathcal{N}(v) \to \{-1,0, 1\}\}$ and $d_H$ be the Hamming distance. We define the set of all estimators that differ from $\sigma_0$ on at most $\eta \log n$ vertices in $\cN(v)$, relative to $\sigma_0(u_0)$, as
  \$
  W'(v; \eta) &= \{\sigma\in W(v)\colon d_H(\sigma(\cdot), \sigma_0(u_0) \sigma_0(\cdot)) \le \eta\log n \} \\
  &= \{\sigma\in W(v)\colon d_H(\sigma_0(u_0)\sigma(\cdot),  \sigma_0(\cdot)) \le \eta\log n \}.\$
  Let $\cE_v$ be the event that there exists $\sigma \in W'(v; \eta)$ such that Poisson testing with respect to $\sigma$ fails on vertex $v$ when $\cE_2$ holds; that is,
  \#\label{eq:def-cE_v}
  \mathcal{E}_v &= \Big[\big\{\sigma_0(v) = 1 \big\} \bigcap \Big(\bigcup_{\sigma \in W'(v;\eta)}\big\{\tau(v, \sigma_0(u_0)\sigma) \leq 0 \big\}\Big)\Big] \nonumber\\
  &\qquad \bigcup \Big[\big\{\sigma_0(v) = -1 \big\} \bigcap \Big(\bigcup_{\sigma \in W'(v;\eta)}\big\{\tau(v, \sigma_0(u_0)\sigma) \geq 0 \big\}\Big)\Big].
  \#
  We provide some insights into the definition of $\cE_v$. Recall that $\sigma_{\textsf{genie}}(v)=\text{sign}(\tau(v,\sigma_0))$ defined in \eqref{eq:genie} picks the event with the larger likelihood between $\{\sigma_0(v)=1\}$ and $\{\sigma_0(v)=-1\}$. Thus, for example, suppose that $\sigma_0(v)=1$, then $\sigma_{\textsf{genie}}(v)$ makes a mistake when $\tau(v,\sigma_0)\le0$.
  We consider any $\sigma\in W'(v;\eta)$. Since $\sigma_0(u_0)\sigma(u)=\sigma_0(u)$ for most $u\in\cN(v)$, $d(v,\sigma_0(u_0)\sigma)$ and $d(v,\sigma_0)$ and thus $\tau(v,\sigma_0(u_0)\sigma)$ and $\tau(v,\sigma_0)$ are close. Formalizing the intuition, suppose that $\sigma_0(v) = 1$. If $\sigma_0(u_0) = 1$, then for $\cE_2$ to hold, we must classify $v$ as $+1$ to be correct relative to $\sigma_0(u_0)$. Thus, $v$ is misclassified relative to $\sigma$ whenever $\tau(v, \sigma) \leq 0$. If $\sigma_0(v) = 1$ and $\sigma_0(u_0) = -1$, then we must classify $v$ as $-1$. Then $v$ is misclassified relative to $\sigma$ whenever $\tau(v, \sigma) \geq 0$. As a summary, failure in the case $\sigma_0(v) = 1$ means $\tau(v,  \sigma_0(u_0)\sigma) \leq 0$.
  
  It follows that 
  \begin{equation}\mathbb{P}\left(\mathcal{E}_2(v)^c \cap \mathcal{E}_1\right) \leq \mathbb{P}(\cE_v). \label{eq:v-failure}
  \end{equation}
  We aim to show that for $\eta > 0$ sufficiently small, $\mathbb{P}(\cE_v) = n^{-(1 + \Omega(1))}$. Due to the uniform prior on $\sigma_0(v)$, we have 
  \#\label{eq:prob_Ev}
  \mathbb{P}(\cE_v) &= \frac{1}{2} \big[\mathbb{P}(\cE_v \mid \sigma_0(v) = 1) + \mathbb{P}(\cE_v \mid \sigma_0(v) = -1)\big].    
  \#
  We now bound the first term in \eqref{eq:prob_Ev}. Let $D \in \mathbb{Z}_+^4$ represent the ground-truth degree profile of vertex $v$. We consider a realization $D=d(v,\sigma_0)$ and the induced $\tau(v,\sigma_0)$. Next, we bound the distance $|\tau(v,\sigma_0(u_0)\sigma) - \tau(v,\sigma_0 )|$ for any $\sigma\in W'(v;\eta)$. We note that the edges and non-edges are fixed in a given graph $G$; that is, for any $\sigma\in W(v)$, we have 
  \$
  &d_{1}^+(u,\sigma_0(u_0)\sigma) + d_{-1}^+(u,\sigma_0(u_0)\sigma) = d_{1}^+(u,\sigma_0) + d_{-1}^+(u,\sigma_0), \\
  & d_{1}^-(u,\sigma_0(u_0)\sigma) + d_{-1}^-(u,\sigma_0(u_0)\sigma) = d_{1}^-(u,\sigma_0) + d_{-1}^-(u,\sigma_0). \$ Let $\alpha = d_{1}^+(u,\sigma_0(u_0)\sigma) - d_{1}^+(u,\sigma_0)=-(d_{-1}^+(u,\sigma_0(u_0)\sigma) - d_{-1}^+(u,\sigma_0) )$ and $\beta=d_{1}^-(u,\sigma_0(u_0)\sigma) - d_{1}^-(u,\sigma_0) = - (d_{-1}^-(u,\sigma_0(u_0)\sigma) - d_{-1}^-(u,\sigma_0) )$. It follows that
  \$
  \tau(v,\sigma_0(u_0)\sigma) - \tau(v,\sigma_0 ) &=\log\big(\frac{1-a}{1-b} \big)[d_{1}^-(u,\sigma_0(u_0)\sigma) - d_{1}^-(u,\sigma_0)- (d_{-1}^-(u,\sigma_0(u_0)\sigma) - d_{-1}^-(u,\sigma_0) )] \\
  & \quad + \log\big(\frac{a}{b} \big)[d_{1}^+(u,\sigma_0(u_0)\sigma) - d_{1}^+(u,\sigma_0)- (d_{-1}^+(u,\sigma_0(u_0)\sigma) - d_{-1}^+(u,\sigma_0) )]\\
  & = 2\Big[\alpha\cdot\log\big(\frac{a}{b} \big) +\beta\cdot \log\big(\frac{1-a}{1-b} \big)\Big].
  \$
  For any $\sigma\in W'(v;\eta)$, recalling that $d_H(\sigma_0(u_0)\sigma(\cdot), \sigma_0(\cdot))\le\eta\log n$, we have $|\alpha|\le\eta\log n$ and $|\beta|\le\eta\log n$. Thus, we define $\rho = 2\cdot[|\log(a/b)| + |\log((1-a)/(1-b))|]$ and have
  \$
  \big|\tau(v,\sigma_0(u_0)\sigma) - \tau(v,\sigma_0 ) \big| \le 2\Big[|\alpha|\cdot\big|\log\big(\frac{a}{b} \big)\big| +\big|\beta\big|\cdot \big|\log\big(\frac{1-a}{1-b} \big)\big|\Big]\le \rho\eta\log n.
  \$  
  
  We define a set $Y\subset \Z_+^4$ as follows:
  \$
  Y = \Big\{d=(d_{1}^+, d_{1}^-, d_{-1}^+, d_{-1}^-)\in \mathbb{Z}_+^4 : \log\big(\frac{a}{b} \big)\big(d_{1}^+ - d_{-1}^+ \big)+\log\big(\frac{1-a}{1-b} \big)\big(d_{1}^- - d_{-1}^- \big) \leq \rho \eta \log n\Big\}. 
  \$
  Conditioned on $\{\sigma_0(v) = 1\}$, Poisson testing fails relative to $\sigma$ when $\tau(v,\sigma_0(u_0)\sigma) \leq 0$. Thus, 
  \begin{align*}
  &\mathbb{P}\big(\cE_v \biggiven \sigma_0(v) = 1\big) \\
  &= \sum_{d\in\Z_+^4}\pr\Big(\{D=d\}\bigcap\big\{\min_{\sigma\in W'(v;\eta)} \tau(v,\sigma_0(u_0)\sigma) \leq 0\big\} \biggiven \sigma_0(v) = 1 \Big) \\
  & \le \sum_{d\in\Z_+^4}\pr\Big(\{D=d\}\bigcap\big\{\tau(v,\sigma_0 ) \le \rho\eta\log n \big\}\biggiven \sigma_0(v) = 1 \Big) \\
  & = \sum_{d \in Y} \mathbb{P}(D = d \mid \sigma_0(v) = 1). 
  \end{align*}
  To bound the above summation, we consider random variables $\widetilde D \sim \text{Poisson} (x)$ with $x$ defined in \eqref{eq:def-x-y} and $X$ defined in \eqref{eq:def_X}.
  Recalling that $D\sim \widetilde D$ conditioned on $\sigma_0(v)=1$, Lemma \ref{lem:bound-X} gives that
  \$
  \mathbb{P}\big(\cE_v \biggiven \sigma_0(v) = 1\big) &\le \sum_{d \in Y} \mathbb{P}\big(D = d \biggiven \sigma_0(v) = 1\big) \\
  &= \mathbb{P}\big(\widetilde D \in Y\big) \\
  &= \pr\big(X\ge -\rho\eta\log n \big)\\
  &\le n^{-\lambda \nu_d(1-\sqrt{ab} - \sqrt{(1-a)(1-b)})+\rho\eta/2}.
  \$
  Since $\lambda \nu_d(1-\sqrt{ab} - \sqrt{(1-a)(1-b)}) > 1$, we take $\eta = (\lambda \nu_d(1-\sqrt{ab} - \sqrt{(1-a)(1-b)}) -1)/\rho>0$ and
  conclude that
  \[\mathbb{P}\big(\cE_v \biggiven \sigma_0(v) = 1\big) \leq n^{-\frac{1}{2}(\lambda \nu_d(1-\sqrt{ab} - \sqrt{(1-a)(1-b)}) +1 )} = o(1/n).\]
  
  Similarly, we study the case conditioned on $\{\sigma_0(v) = -1\}$. Let $Y' = \{d=(d_{1}^+, d_{1}^-, d_{-1}^+, d_{-1}^-)\in \mathbb{Z}_+^4 : \log(a/b)(d_{1}^+ - d_{-1}^+ )+\log((1-a)/(1-b) )(d_{1}^- - d_{-1}^- ) \ge -\rho \eta \log n\}$. The definition of $\cE_v $ in \eqref{eq:def-cE_v} gives that
  \$
  &\mathbb{P}\big(\cE_v \biggiven \sigma_0(v) = -1\big) \\
  &= \sum_{d\in\Z_+^4}\pr\Big(\{D=d\}\bigcap\big\{\max_{\sigma\in W'(v;\eta)} \tau(v,\sigma_0(u_0)\sigma) \ge 0\big\} \biggiven \sigma_0(v) = -1 \Big) \\
  & \le \sum_{d\in\Z_+^4}\pr\Big(\{D=d\}\bigcap\big\{\tau(v,\sigma_0 ) \ge -\rho\eta\log n \big\}\biggiven \sigma_0(v) = -1 \Big) \\
  & = \sum_{d \in Y'} \mathbb{P}(D = d \mid \sigma_0(v) = -1). 
  \$
  For the same $\widetilde D=[D_1^+, D_1^-, D_{-1}^+, D_{-1}^-] \sim \text{Poisson} (\lambda \nu_d \log n [a, 1-a,b,1-b ]/2 )$, note that condition on $\sigma_0(v) = -1$, we have $D\sim [D_{-1}^+, D_{-1}^-,D_1^+, D_1^-]$.
  Thus, with the same $X$ defined in \eqref{eq:def_X}, we have \$
  \sum_{d \in Y'} \mathbb{P}(D = d \mid \sigma_0(v) = -1) &= \pr\big([D_{-1}^+, D_{-1}^-,D_1^+, D_1^-]\in Y'\big) \\
  & = \pr\Big(\log\big(\frac{a}{b} \big)\big(D_{-1}^+ - D_{1}^+ \big) + \log\big(\frac{1-a}{1-b} \big)\big(D_{-1}^- - D_{1}^- \big) \ge -\rho\eta\log n\Big) \\
  & = \pr(X\ge -\rho\eta\log n).
  \$
  Thus, similarly, Lemma \ref{lem:bound-X} gives that $\pr(\cE_v \given \sigma_0(v) = -1) \leq n^{-\frac{1}{2}(\lambda \nu_d(1-\sqrt{ab} - \sqrt{(1-a)(1-b)}) +1 )}$. 
  Therefore, the above bound together with \eqref{eq:union-bound}, \eqref{eq:v-failure}, and \eqref{eq:prob_Ev} implies $\pr(\cE_2^c \cap \cE_1 \cap \cE_0) = o(1)$. 
  Finally, we have $\pr((\cE_2')^c)\le\pr(\cE_2^c)=o(1)$ due to \eqref{eq:phase-II-failure}.
  \end{proof}
  
  \section{Impossibility: Proof of Theorem \ref{theorem:impossibility-general}}\label{sec:impossibility}
  In this section, we prove the impossibility of exact recovery under the given conditions and complete the proof of Theorem \ref{theorem:impossibility-general}. Recalling that Theorem \ref{theorem:impossibility} (Theorem 3.7 in \cite{Abbe2021}) has already established the impossibility when $\lambda>0$, $d\in\N$, and $0\le b< a\le 1$ satisfying \eqref{eq:IT-threshold}. Here, we extend the same result to the case where the requirement $a>b$ is dropped.
  \begin{proposition}\label{prop:impossible}
  Let $\lambda > 0$, $d \in \mathbb{N}$, and $a,b \in [0,1]$ satisfy \eqref{eq:IT-threshold} and let $G_n \sim \text{GSBM}(\lambda, n, a, b, d)$. Then any estimator $\widetilde{\sigma}$ fails to achieve exact recovery.     
  \end{proposition}
  \begin{proof}
  We note that the analysis of Theorem \ref{theorem:impossibility} builds upon Lemma 8.2 in \cite{Abbe2021}, which itself relies on Lemma 11 from \cite{Abbe2015}. Lemma 11 provides the error exponent for hypothesis testing between Poisson random vectors, forming the basis for the impossibility result. Notably, only the CH-divergence criterion $\lambda \nu_d(1-\sqrt{ab} - \sqrt{(1-a)(1-b)})<1$ is needed to ensure the indistinguishability of the two Poisson distributions. Therefore, the impossibility in Theorem \ref{theorem:impossibility} can be readily extended to the case where the condition $a>b$ is dropped.
  \end{proof}
  
  Moreover, we show the impossibility of exact recovery for $d = 1$ and $\lambda<1$.
  \begin{proposition} \label{prop:lambda<1}
      When $d=1$, let $0<\lambda < 1$ and $a,b \in [0,1]$ and let $G_n \sim \text{GSBM}(\lambda, n, a, b, d)$. Then any estimator $\widetilde{\sigma}$ fails to achieve exact recovery. 
  \end{proposition}
  \begin{proof}
      When $d=1$, we partition the interval $[-n/2, n/2]$ into $n/\log n$ blocks of length $\log n$ each. Notably, if there are $k\ge 2$ mutually non-adjacent empty blocks, the interval gets divided into $k\ge 2$ disjoint segments that lack mutual visibility. In such scenarios, achieving exact recovery becomes impossible as we can randomly flip the signs of one segment. Formally, suppose that there are $k$ segments, where the $i$th segment contains blocks $\{B_{j}\colon j \in \text{seg}(i)\}$ for $\text{seg}(i) \subset [n/\log n]$.
      Then for any $s \in \{\pm 1\}^k$, the labeling $\sigma_0$ has the same posterior probability as $\sigma(\cdot; s)$, defined as
      \[\sigma(v;s) = \sigma_0(v) \sum_{i \in [k]}  s_i \sum_{j \in \text{seg}(i)} \mathds{1}_{\{v \in B_j\}}.\]
      It follows that the error probability of the genie-aided estimator is at least $1 - 2/{2^k} = 1 - 1/{2^{k-1}}$, conditioned on there being $k$ segments. Let $\cX$ be the event of having at least two non-adjacent empty blocks (and thus two segments). The aforementioned observation means that if $\cX$ holds, the error probability is at least $1/2$, and thus the exact recovery is unachievable. 
  
      We now prove that $\pr(\cX)=1-o(1)$ if $\lambda<1$. Let $\cY_k$ be the event of having exactly $k$ empty blocks, among which at least two of them are non-adjacent. Recalling that each block is independently empty with probability $\exp(-\lambda \log n) = n^{-\lambda}$, we have 
      \$
          \pr(\cX) &= \sum_{k=2}^{n/\log n} \pr(\cY_k) = \sum_{k=2}^{n/\log n - 1} \Big(\binom{{n}/{\log n}}{k} - {n}/{\log n}\Big) \big(n^{-\lambda}\big)^k\big(1 - n^{-\lambda}\big)^{n/\log n-k} \\
          & \ge \sum_{k=1}^{n/\log n} \binom{{n}/{\log n}}{k} (n^{-\lambda})^k(1 - n^{-\lambda})^{n/\log n-k} - \frac{n}{\log n}(1 - n^{-\lambda})^{n/\log n}\sum_{k=1}^{n/\log n} 
          \big[n^{-\lambda}/(1 - n^{-\lambda})\big]^{k} \\
          & \ge 1- \big(1 - n^{-\lambda}\big)^{n/\log n} - \big(1 - n^{-\lambda}\big)^{n/\log n} \cdot\frac{n}{\log n} \cdot\frac{n^{-\lambda}}{1-2n^{-\lambda}} \\
          & \ge 1 - \big(1 - n^{-\lambda}\big)^{n/\log n}\cdot\big(1 + 2n^{1-\lambda}/\log n\big) \\
          & = 1 - \big[(1 - n^{-\lambda})^{n^{\lambda}}\big]^{n^{1-\lambda}/\log n}\cdot\big(1 + 2n^{1-\lambda}/\log n\big) \\
          &= 1 - O\big(\exp(-{n^{1-\lambda}/\log n})\cdot(1 + 2n^{1-\lambda}/\log n)\big) =  1- o(1),
      \$
  where the second inequality follows by calculating the Binomial series and the geometric series, and the last inequality holds since $1-2n^{-\lambda}\ge1/2$ for large enough $n$.
  \end{proof}
  In summary, by combining Propositions \ref{prop:impossible} and \ref{prop:lambda<1}, we complete the proof of Theorem \ref{theorem:impossibility-general}.
  
  \section{Further related work}\label{sec:related-work}
  Our work contributes to the growing literature on community recovery in random geometric graphs, beginning with latent space models proposed in the network science and sociology literature (see for example \cite{Handcock2007,Hoff2002}). There have been several models for community detection in geometric graphs. The most similar to the one we study is the Soft Geometric Block Model (Soft GBM), proposed by Avrachenkov et al \cite{Avrachenkov2021}. The main difference between their model and the GSBM is that the positions of the vertices are unknown. Avrachenkov et al \cite{Avrachenkov2021} proposed a spectral algorithm for almost exact recovery, clustering communities using a higher-order eigenvector of the adjacency matrix. Using a refinement procedure similar to ours, \cite{Avrachenkov2021} also achieved exact recovery, though only in the denser linear average degree regime.  
  
  A special case of the Soft GBM is the Geometric Block Model (GBM), proposed by Galhotra et al \cite{Galhotra2018} with follow-up work including \cite{Galhotra2022,Chien2020}. In the GBM, community assignments are generated independently, and latent vertex positions are generated uniformly at random on the unit sphere. Edges are then formed according to parameters $\{\beta_{i,j}\}$, where pair of vertices $u,v$ in communities $i,j$ with locations $Z_u, Z_v$ are connected if $\langle Z_u, Z_v \rangle \leq \beta_{i,j}$.
  
  In the previously mentioned models, the vertex positions do not depend on the community assignments. In contrast, Abbe et al \cite{Abbe2020b} proposed the Gaussian-Mixture Block Model (GMBM), where (latent) vertex positions are determined according to a mixture of Gaussians, one for each community. Edges are formed between all pairs of vertices whose distance falls below a threshold. A similar model was recently studied by Li and Schramm \cite{Li2023} in the high-dimensional setting. Additionally, P\'ech\'e and Perchet \cite{peche2020robustness} studied a geometric perturbation of the SBM, where vertices are generated according to a mixture of Gaussians, and the probability of connecting a pair of vertices is given by the sum of the SBM parameter and a function of the latent positions. 
  
  In addition, some works \cite{araya2019latent,eldan2022community} consider the task of recovering the geometric representation (locations) of the vertices in random geometric graphs as a form of community detection. Their setting differs significantly from ours. We refer to the survey \cite{duchemin2023random} for an overview of the recent developments in non-parametric inference in random geometric graphs.

  \section{Conclusions and future directions}\label{sec:future-directions}
  Our work identifies the information-theoretic threshold for exact recovery in the two-community, balanced, symmetric GSBM. A natural direction for future work is to consider the case of multiple communities, with general community membership probabilities and general edge probabilities. We believe that the information-theoretic threshold will again be given by a CH-divergence criterion, and a variant of our two-phase approach will achieve the threshold. 
  
  It would also be interesting to study other spatial network inference problems. For example, consider $\mathbb{Z}_2$-synchronization \cite{Bandeira2017,Javanmard2016,Abbe2020}, a signal recovery problem motivated by applications to clock synchronization \cite{giridhar2006distributed}, robotics \cite{Rosen2020}, and cryogenic electron microscopy \cite{Singer2011}. In the standard version of the problem, each vertex is assigned an unknown label $x(v) \in \{\pm 1\}$. For each pair $(u,v)$, we observe $x(u) x(v) + \sigma W_{uv}$, where $\sigma > 0$ and $W_{uv} \sim \mathcal{N}(0,1)$. 
  Now suppose that the vertices are generated according to a Poisson point process, and we observe $x(u) x(v) + \sigma W_{uv}$ only for mutually visible vertices, which models a signal recovery problem with spatially limited observations. An open question is then whether our two-phase approach can be adapted to this synchronization problem.

\vskip6pt
\paragraph*{Acknowledgements.} J.G. was supported in part by NSF CCF-2154100. X.N. and E.W. were supported in part by NSF ECCS-2030251 and CMMI-2024774.

\bibliography{references}
\bibliographystyle{abbrv}
\end{document}